\documentclass[12pt]{article}
\usepackage{amsmath, amsthm, amssymb}
\usepackage{graphicx}
\usepackage{enumerate}
\usepackage{natbib}
\usepackage{graphicx}
\usepackage{subcaption}
\usepackage{url} 
\usepackage{amsthm}
\usepackage{xcolor}
\usepackage[colorlinks=false,hidelinks]{hyperref}

\usepackage{misc_macros}

\newcommand{\blind}{0}

\graphicspath{{figures/}}

\begin{document}

\def\spacingset#1{\renewcommand{\baselinestretch}%
{#1}\small\normalsize} \spacingset{1}

\newcommand\numberthis{\addtocounter{equation}{1}\tag{\theequation}}


\if0\blind
{
  \title{\bf The Network Structure of Unequal Diffusion}
  \author{
    Eaman Jahani
    \thanks{Authors would like to thank Matthew Jackson for his valuable comments and feedback and Tiago Peixoto for his helpful insights and references.} \\
    Department of Statistics, University of California, Berkeley
    \vspace{0.5em} \\
    Dean Eckles \\
    MIT Sloan School of Management
    \vspace{0.5em} \\
    Alex ``Sandy'' Pentland \\
    MIT Institute for Data, Systems and Society}
  \date{September 15, 2022}
  \maketitle
} \fi

\if1\blind
{
  \bigskip
  \bigskip
  \bigskip
  \begin{center}
    {\LARGE\bf The Network Structure of Unequal Diffusion}
\end{center}
  \medskip
} \fi

\bigskip
\begin{abstract}
Social networks affect the diffusion of information, and thus have the potential to reduce or amplify inequality in access to opportunity.
We show empirically that social networks often exhibit a much larger potential for unequal diffusion across groups along paths of length 2 and 3 than expected by our random graph models. We argue that homophily alone cannot not fully explain the extent of unequal diffusion and attribute this mismatch to unequal distribution of cross-group links among the nodes. Based on this insight, we develop a variant of the stochastic block model that incorporates the heterogeneity in cross-group linking. The model provides an unbiased and consistent estimate of assortativity or homophily on paths of length 2 and provide a more accurate estimate along paths of length 3 than existing models. We characterize the null distribution of its log-likelihood ratio test and argue that the goodness of fit test is valid only when the network is dense. Based on our empirical observations and modeling results, we conclude that the impact of any departure from equal distribution of links to source nodes in the diffusion process is not limited to its first order effects as some nodes will have fewer direct links to the sources.
More importantly, this unequal distribution will also lead to second order effects as the whole group will have fewer diffusion paths to the sources. 
\end{abstract}

\noindent%
\vfill

\newtheorem{lemma}{Lemma}
\newtheorem{proposition}{Proposition}
\newtheorem{theorem}{Theorem}

\newpage
\spacingset{1.5} 
\section{Introduction}
\label{sec:intro}
Diffusion of information in social networks determines who gets access to a valueable piece of information, such as a new investment opportunity. The structure of the network plays an important role in which individuals or groups receive the valuable information. Certain network structures are more likely to keep a piece of information exclusive to one group, thus leading to unequal diffusion. For example, if there are very few social links between people of different races, the information about a new employment opportunity that is generated among one race might never reach individuals of the other race \citep{jackson2004}. Many existing network models aim to explain the absence of diffusion from one group to another through assortative mixing \citep{Newman2003}. Assortative mixing, or simply assortativity, captures the bias in forming edges with similar characteristics. It is also referred to as homophily which simply means that attributes of nodes are correlated across the edges. For example, in social networks individuals have a strong tendency to form links with other people who are similar to them in terms of age, language, socioeconomic status or race.

The stochastic block model (SBM) --- along with its variants such as degree-correction \citep{Karrer2011} --- defines  an important class of these models that explicitly account for assortative mixing in networks. SBM is a generative random network model for modeling blocks or groups in networks. It has been widely used in computer science and social sciences to model community structure in networks \citep{rohe2011spectral,holland1983stochastic,anderson1992building,faust1992blockmodels,wasserman1989canonical,wang1987}. In its original form, vertices in a network exclusively belong to one of the $K$ groups (or blocks) in the network. Each pair of vertices form an edge independently of other edges or vertices. Edge formations between any pairs of two groups are independent, identical and solely determined by the group membership of the pair of vertices. If $g_i \in \{1,2,...,K\}$  corresponds to the group of vertex $i$, then a $K \times K$ matrix, $P$, determines the edge formation probabilities between any pair of vertices. The probability of an edge between any pairs $i$ and $j$ is the  $(g_i, g_j)$ element in the matrix, $P_{g_i, g_j}$.

This simple model can produce a variety of interesting network structures. For example, an edge probability matrix in which diagonal entries are much larger than off-diagonal entries produces networks with densely connected groups and sparse connections across groups. The ability to model such community structure is the main reason SBM can capture assortative mixing in a network. This has led to the popularity of SBM as one of the main methods for community detection. SBM does so by generating random networks that match the observed network in terms of the frequency of within-group and cross-group edges. The fitted model matches the observed assortativity or homophily in expectation.

\begin{figure}[t]
\centering
\begin{subfigure}[t]{0.45\textwidth}
\centering
\includegraphics[width=\textwidth]{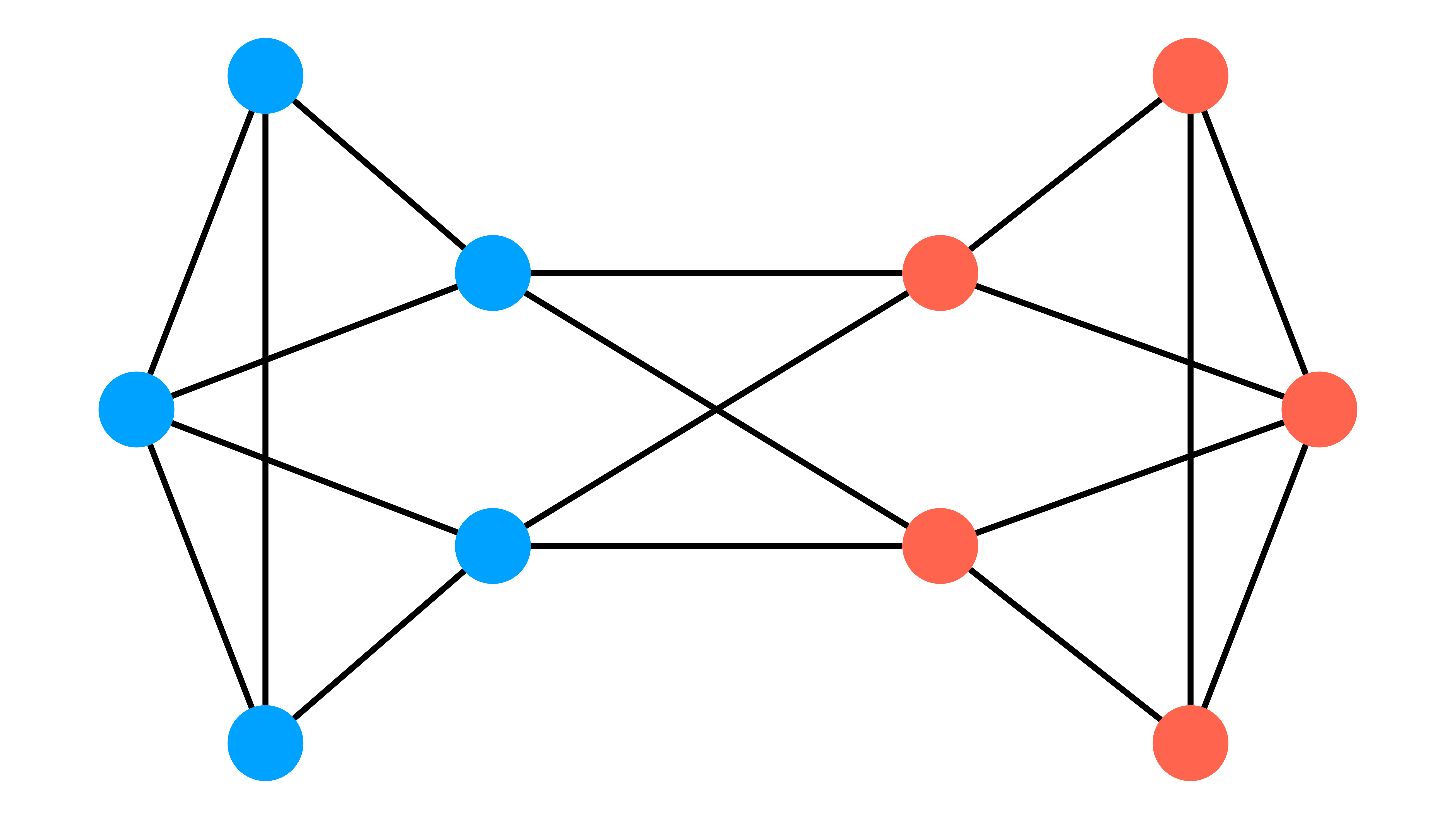}
\caption{}
\label{fig:brokerage_network}
\end{subfigure}
\hfill
\begin{subfigure}[t]{0.45\textwidth}
\centering
\includegraphics[width=\textwidth]{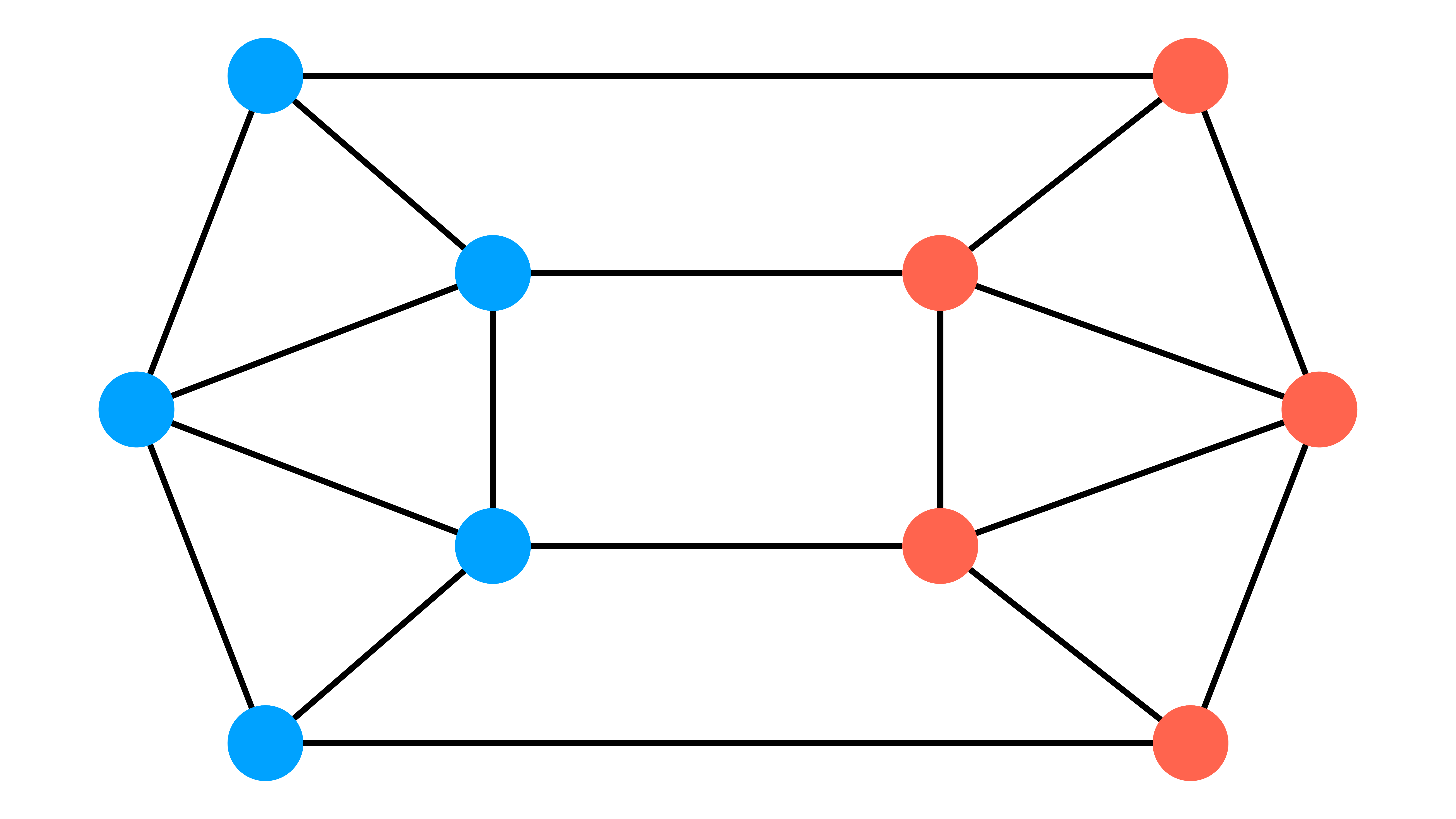}
\caption{}
\label{fig:nonbrokerage_network}
\end{subfigure}
\caption{Comparison of  (\protect\subref{fig:brokerage_network}) a network with brokerage in which a disproportionate fraction of cross-type edges are held by a small number of nodes versus (\protect\subref{fig:nonbrokerage_network}) a similar network in which cross-type edges are distributed more equally. Red and blue nodes correspond to two different groups or blocks. Corresponding nodes have the same degree and the number of cross-type edges are the same in both networks, but there are 10 cross-type paths of length 2 of the form red-blue-blue in network (\protect\subref{fig:nonbrokerage_network}) while there are only 8 such paths in network (\protect\subref{fig:brokerage_network}).}
\label{fig:brokerage_network_example}
\end{figure}

SBM or its degree corrected version assume that within-group and cross-group edges are distributed ``uniformly'' across all pairs: the existence of an edge between any two pairs is identical to other similar pairs. In the case of degree-corrected SBM (DCSBM), after conditioning on degree two nodes are similar in terms of their cross-group edge formation.
In reality, many real networks have heterogeneous propensities in edge formation to various groups. In most cases, social networks exhibit a pattern of brokerage which means cross-group edges are not distributed uniformly, instead a small subgroup of nodes hold a disproportionate level of cross-group edges. \citet{simmel1950} was the first to introduce the concept of network brokerage in triadic relations. \citet{burt2009} later advanced our understanding of brokerage by introducing the concept of ``structural holes'' between two unconnected communities, across which brokers act as intermediary. These broker nodes play an important role in connecting otherwise disconnected communities, moving information between them, and acting as an intermediary for resource exchanges. Due to their unique position in the network, brokers benefit from various types of advantages, for example access to diverse information or opportunities for arbitrage in exchanges. However, these advantages to brokers might lead to some  costs to other actors in the network or the network as a whole.

Figure \ref{fig:brokerage_network} provides a visual example of a network with brokerage in which a small number of broker nodes have a higher propensity to form links with brokers of the out-group, hence maintaining majority of cross-group edges. Figure \ref{fig:nonbrokerage_network} shows a similar network with less brokerage which has more frequent cross-type paths of length 2 even though it has the same degree distribution as the brokerage network \ref{fig:brokerage_network}. While brokers play an integral role in connecting otherwise disconnected communities,  they can nevertheless act a bottleneck by reducing the number of possible paths between any two groups when compared to a similar network with cross-group ties uniformly distributed across the network. Because brokers hold a disproportionate number of cross-group ties,  they can constrain diffusion of information from one group to another. In this paper, we argue that one needs to not only look at homophily or assortativity on paths of length 1, but also on the extent of assortativity of all possible diffusion paths of varying lengths to completely account for unequal diffusion in networks.
We then attempt to incorporate the heterogeneity in edge propensities and in particular brokerage into class of Stochastic Block Models and show that by doing so the model better explains unequal diffusion of information.

We show that while directly fitting for assortativity on paths of length 1, SBM fails to accurately capture assortativity on longer paths in real world networks.
In the context of random graphs, network  brokerage  occurs  when  a  few  nodes  in  the  network  have  higher  probability  to connect with an out-group than other in-group nodes. By incorporating this heterogeneity into our models of random network and in particular SBM or degree-corrected SBM, we show that the generative model can better match the assortativity along longer paths and more generally cross-type diffusion in the observed network.
In section \ref{sec:background}, we discuss SBM and some variant models and show that they consistently under-estimate the observed assortativity on paths of length 2 in 56 school networks, even though these models explicitly accounts for assortativity on paths of length 1.
In section \ref{sec:model}, we discuss a general framework for Stochastic Block Models and develop variants which account for node heterogeneity in brokerage and by doing so match assortativity on paths of length 1 and 2 in expectation.  In section \ref{sec:results}, we provide the results from fitting the school networks to our model and show that  even though not explicitly modeled for, it closely matches assortativity on paths of length 3. In section \ref{sec:model_selection}, we address the goodness of fit for this new model versus one that does not account for brokerage. We characterize the distribution of the log likelihood ratio statistic and argue that the test is valid only if the network is dense, which is often not the case for social networks.

In the remainder of this paper, we mostly focus on assortativity of path length 2 and 3 as opposed to longer paths. While diffusion as a general process can occur across paths of any length, nevertheless in many scenarios, especially those that involve access a valuable resource, diffusion mostly occurs along short paths. Therefore, while assortativities along paths of length 2 and 3 do not provide an exact representation of \textbf{diffusion assortativity}, we believe they nevertheless provide a simple and interpretable model that is applicable in most social contexts.

\section{Background}
\label{sec:background}
\subsection{Assortativity}
Before discussing the Stochastic Block Model and its properties regarding diffusion, we need to explain the assortativity coefficient, a common way to quantify the level of assortative mixing in a network.
The assortativity coefficient in a directed network, which quantifies the bias in favor of edges between in-group nodes, is defined as below \citep{newman2003mixing}.
\begin{align}
    r^{(1)} = \frac{\sum \limits_{r} e_{rr} - \sum \limits_{r} a_r b_r}{1 - \sum \limits_{r}a_r b_r}
    \label{eq:assort1}
\end{align}
where the quantity $e_{rs}$ is the fraction of total (directed) edges from a node in group $r$ to a node in group $s$, $a_r$ is the fraction of total edges from a node in group $r$ and $b_r$ is the fraction of total edges to a node in group $r$. Below we denote the adjacency matrix as $\mathbf{A}$ and the group of node $i$ as $g_i$.
\begin{align}
    e_{rs} = \frac{\sum \limits_{i,j} A_{ij} \delta_{g_i,r} \delta_{g_j,s}}{\sum \limits_{i,j} A_{ij}} \quad \quad \quad
    a_r = \sum_s e_{rs} \quad \quad \quad
    b_r = \sum_s e_{sr}
    \label{eq:assort1_helper}
\end{align}
The numerator in equation \ref{eq:assort1} is simply the \textit{modularity} of the network, another quantity for the strength of community structure in networks \citep{newman2006modularity, newman2004finding, geng2019probabilistic} that measures the fraction of in-group edges minus its expected value if the stubs were randomly rewired. The assortativity coefficient is effectively the scaled modularity such that $-1 \le r^{(1)} \le 1$. The $(1)$ superscript in equation \ref{eq:assort1} indicates assortativity is measured on paths of length 1.

\subsection{Assortativity on Longer Paths}
We can define higher order measures of assortativity to quantify the level of assortative mixing along diffusion paths. For example, to compute assortativity on paths of length 2 on a (directed) network, we first construct its corresponding network along paths of length 2 forbidding the traversal of the same edge multiple times and call it the second order network. In this network, there is a (directed) edge from node $i$ to $j$ for every path of length 2 from $i$ to $j$ in the original network.
The assortativity of the second order network corresponds to assortativity along paths of length 2 in the original network denoted by $r^{(2)}$. The second order network will be a multi-graph with potential self-loops, both of which are compatible with the definition of assortativity in equation \ref{eq:assort1}. A similar measure to assortativity on longer paths, but in terms of degree assortativity, is discussed in \citep{arcagni2017higher}.

\subsection{Stochastic Block Model}
The Stochastic Block Model (SBM) \citep{Holland1983} is a simple random network model that allow for communities and heterogeneous edge formation between them. It assumes edge formation between a pair of nodes solely depends on their observed block membership and is independent of other pairs. Consequently, all nodes within a block in SBM have the same binomial distribution for their in-group and out-group degree. Often, the SBM is characterized with a matrix whose elements determine the probability of an edge between any pair of blocks. For example, if we assume two groups in the network, the probability matrix for the undirected SBM has the following form.
\begin{align}
P = 
    \begin{bmatrix}
    p_{11} & p_{12} \\
    p_{12} & p_{22}
    \end{bmatrix}
\end{align}

An appealing property of SBM is that it accurately captures the strength of community structure or assortative mixing in a network. In particular, if we let $\hat{r}$ denote the assortativity coefficient of a sampled network from the maximum likelihood fit, $\widehat{P}$, we have the following convergence in probability as network size grows.
\begin{align}
\hat{r}^{(1)} \xrightarrow{\;\; p \;\;} r^{(1)}
\label{eq:convergence_assort}
\end{align}
In fact, if the network is large enough it can be shown that assortativity from the fitted MLE model approximately matches the observed assortativity in expectation, with exact equality in the case of microcanonical SBM \citep{peixoto2017nonparametric}:
\begin{align}
\EE{\hat{r}^{(1)}} \approx r^{(1)}.
\label{eq:expecation_assort}
\end{align}

Despite its simplicity and its wide-spread use to model community structure, SBM has serious drawbacks when it is used to model real-world networks. The main problem with SBM is its inability to allow for degree heterogeneity within a block. This makes SBM an unreasonable model in real world networks which exhibit high levels of degree heterogeneity \citep{peixoto2015}. A maximum likelihood fitting procedure as described above, in the presence of degree heterogeneity, results in communities of high and low degree nodes. In particular, the maximum likelihood estimate captures degree heterogeneity rather than actual community structure since it splits nodes from the same block into distinct blocks differentiated by their degree. For example, \citet{Bickel2009} showed that SBM splits nodes in the famous Karate club network according to their degree rather than extracting the actual communities. 

To avoid this problem, the degree-corrected SBM \citep[DCSBM;][]{Karrer2011} modifies the generative model such that nodes can have different degrees in each block. It does so by introducing a degree-correction parameters for each node that simulates the node's propensity to form edges, hence controlling for the expected degree of each node separately.
A node with a larger value of degree-correction parameter is expected to have larger degree than a node with smaller value and in the same block. Furthermore similar to SBM, the degree-corrected SBM has additional parameters that control for the propensity of any two groups to form links independent of each node's individual degree propensity. SBM is a special cases of its degree-corrected SBM (DCSBM) when all node degree parameters within a single block are equal. Similar to SBM, the fitted maximum likelihood model for DCSBM also matches the observed assortativity as expressed in equations \ref{eq:convergence_assort} and \ref{eq:expecation_assort}.

Despite its ability to model for degree heterogeneity and its success in real world problems, DCSBM is unable to model heterogeneity in in-group and out-group propensities or brokerage since it uses a single parameter per pair of blocks as their edge propensity. In other words, conditional on total degree, all nodes within a block have the same in-group and out-group degree distribution. This makes it difficult for DCSBM to accurately capture assortativity on longer paths if the network exhibits brokerage, as discussed above and shown below empirically. The DCSBM maximum likelihood estimates underestimate higher order assortativity, even though the expected assortativity on paths of length 1 from a DCSBM maximum likelihood fit matches its observed value.

\subsection{Empirical Study of Higher Order Assortativities with DCSBM}
\label{sec:dcsbm-empirical}
In this section, we analyze a collection of real-world social networks and show that many have assortativity on paths of length 2 that is not predicted by SBM which explicitly fits assortativity on paths of length 1. We reuse the data already collected from a previous study that fully mapped out the social network in 56 middle schools \citep{paluck}. These networks are directed and as such we fit them to a directed DCSBM model. We use these networks to study how and whether DCSBM models mixing structure and in particular higher order assortativity accurately. The data also contains various attributes, such as gender, grade, age and GPA per each student. We will use these attributes to define subgroups within the school network and measure the extent of assortativity on paths of length 1, 2 and 3 along several subgroup characterization.

\begin{figure}
     \centering
     \begin{subfigure}[b]{\textwidth}
         \centering
         \includegraphics[width=0.47\textwidth]{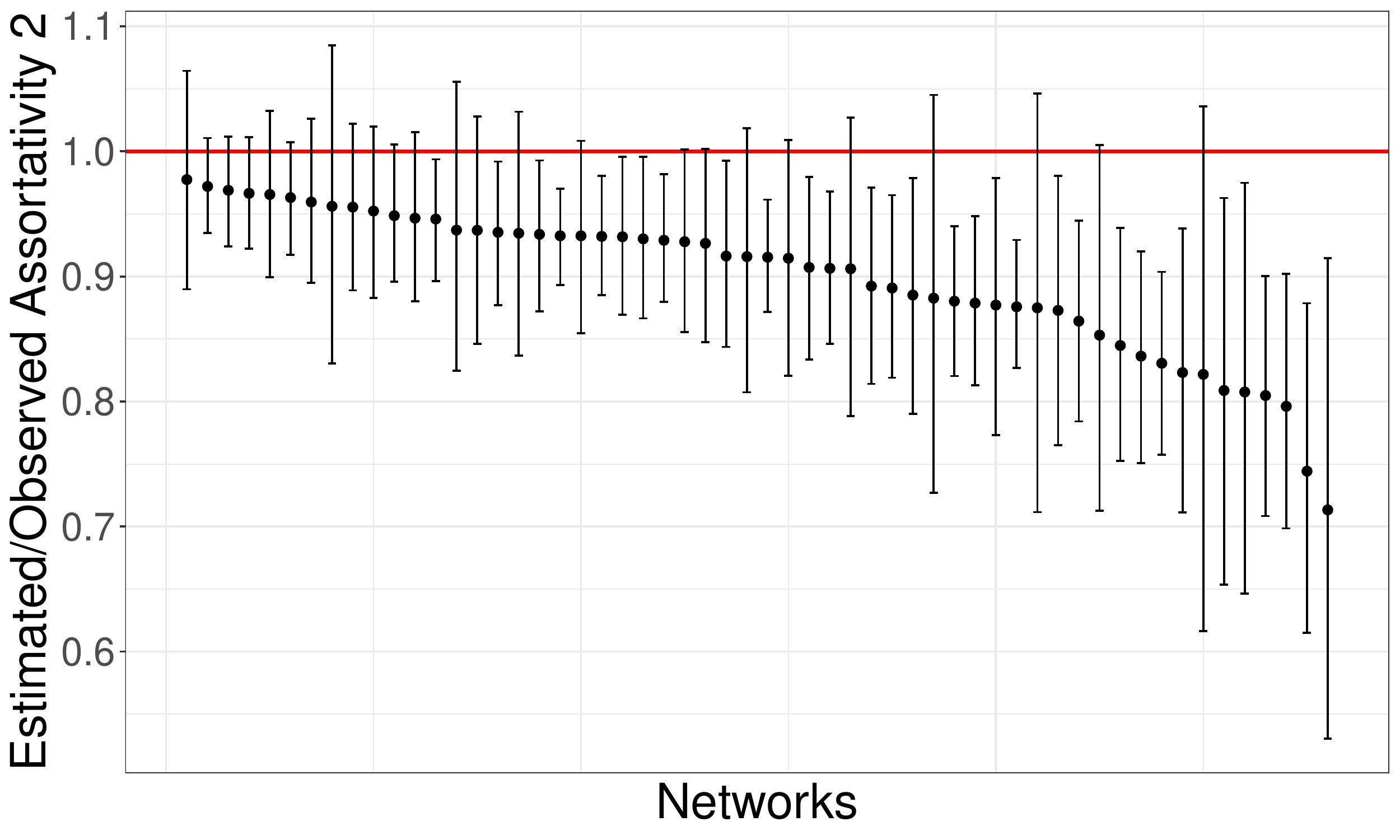}
         \hspace{0.04\textwidth}
         \includegraphics[width=0.47\textwidth]{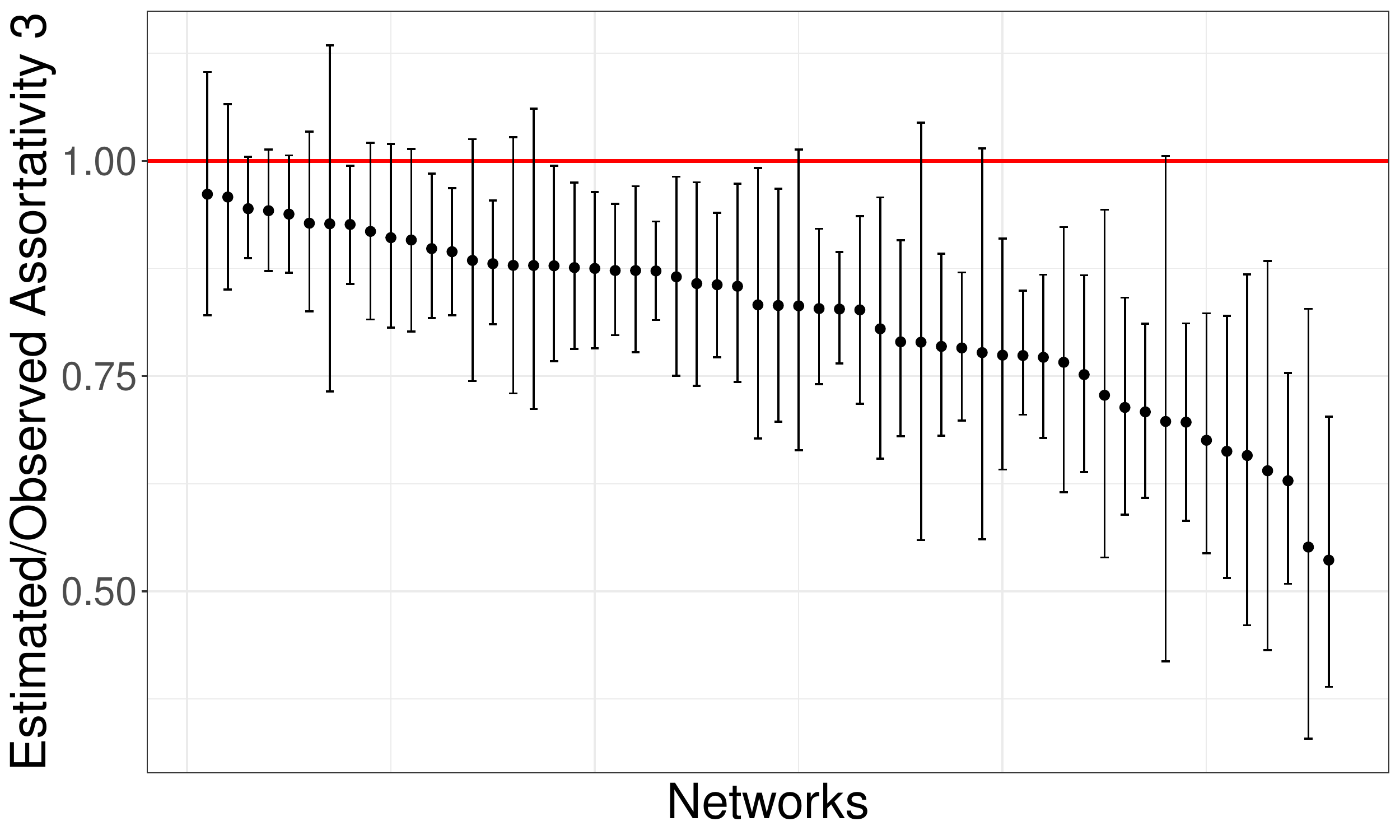}
         \caption{Gender}
     \end{subfigure}
     \par\bigskip 
     \begin{subfigure}[b]{\textwidth}
         \centering
         \includegraphics[width=0.47\textwidth]{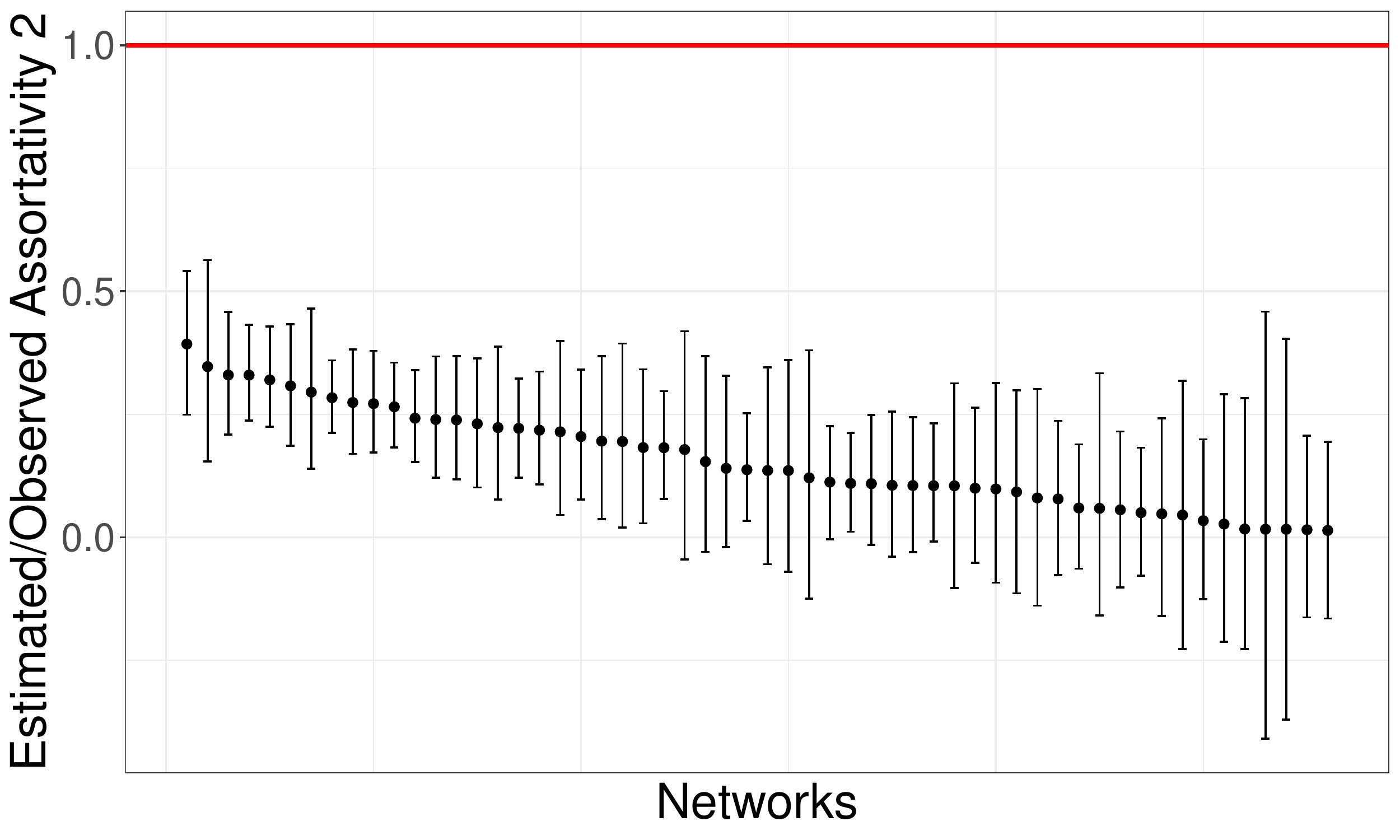}
         \hspace{0.04\textwidth}
         \includegraphics[width=0.47\textwidth]{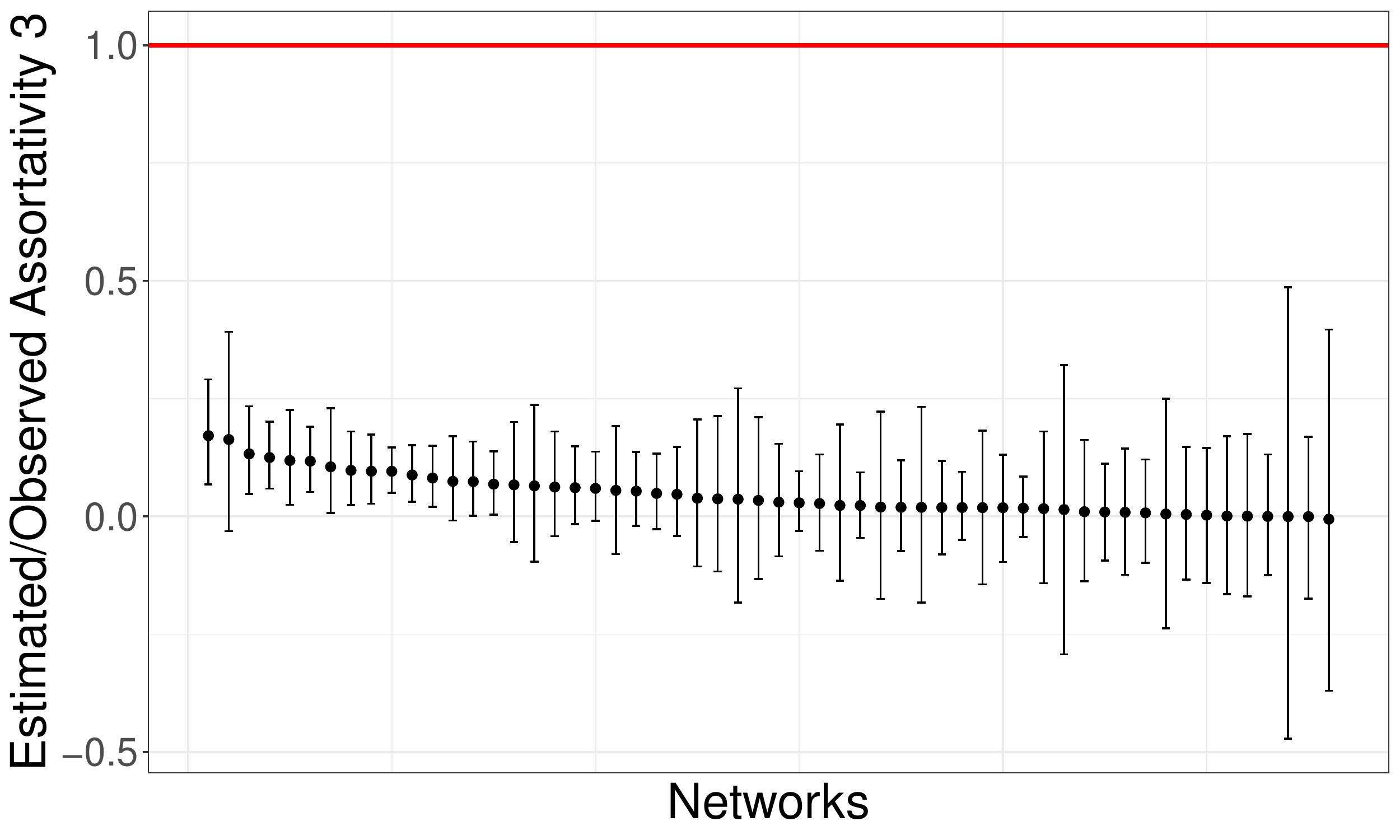}
         \caption{Race}
     \end{subfigure}
    \caption{The distribution of predicted over observed ratio of assortativities on paths of length 2 (left column), $\frac{\hat{r}^{(2)}}{r^{(2)}}$, and 3 (right column), $\frac{\hat{r}^{(3)}}{r^{(3)}}$, from DCSBM along gender (top row) and racial (bottom row) groups. Bars correspond to 95\% confidence interval and each bar corresponds to one school network. Networks are sorted in descending order of the point estimate.}
    \label{fig:dcsbm_observed_vs_predicted_assorts}
\end{figure}

Given the maximum likelihood fit to an observed network, we can generate the distribution of higher order assortativities in a Monte Carlo fashion through repeated sampling of networks from the fitted model, $\widehat{P}$, and computing their assortativity along paths of length 2. This re-sampling procedure to compare other statistical and topological properties of the simulated network, not explicitly accounted for in the model, with the observed network has been used in previous works \citep{williams2000simple, clauset2008hierarchical, foster2011clustering, fischer2015sampling}. This procedure is similar to posterior predictive checks in the Bayesian context \citep{gelman1996posterior}, and it can be used to evaluate the fitness of a model beyond the scope it was designed for.
In our case,  this process reveals that the observed assortativities on paths of length 2 and 3 among the 56 schools are consistently higher than their expected distribution by DCSBM, among all grouping attributes. 
For example, the DCSBM fit based on gender matches the observed assortativity on paths of length 1 in expectation, but 31 out of 56 schools (55\%) exhibit higher assortativity on paths of length 2 than predicted by the fitted model, with two-tailed p-values less than 0.05. Similarly, DCSBM fit based on gender-grade groups (up to 6 groups) leads to 43 schools (76\%) with significantly higher ($p < 0.05$) assortativity on paths of length 2 than predicted by the model.

Figure \ref{fig:dcsbm_observed_vs_predicted_assorts} compares the observed assortativity on paths of length 2 and 3 based on both gender and race (encoded as majority or other) groups with the estimated value from the maximum likelihood model. Even though the observed assortativity on paths of length 1 is always covered by its 95\% confidence interval and very close to the point estimate, the fitted models consistently underestimate higher order assortativities.
The model under-estimates assortativity on paths of length 3 even more than paths of length 2.
Furthermore, DCSBM becomes more inaccurate at predicting higher order assortativities at smaller values. For example, racial assortativity (on paths of length 1) in the schools ranges from 0.02 to 0.32 as opposed to gender which ranges from 0.43 to 0.83, and figure \ref{fig:dcsbm_observed_vs_predicted_assorts} shows that the scale of underestimation is larger for race than gender.

A possible explanation for these discrepancies is the unequal distribution of cross-group edges in the observed networks, while SBM assumes uniform distribution of cross-group edges among all pairs. High brokerage in a network would suggest that a small fraction of nodes in each group hold a large fraction of out-group edges. Conditioned on degree, a more equal distribution of cross-type edges would create extra paths of length 2, thus reducing higher order assortativities.

\section{Model}
\label{sec:model}
In this section, we describe our model that accounts for heterogeneity in out-group edge formation or brokerage and by doing so provides a more accurate estimate of higher order assortativities. Before explaining the model, we restate important concepts and notations used in the model.
\subsection{Preliminary}
\label{sec:model_prelim}
\textbf{Higher Order Networks:} Given a network $W$, its k\textsuperscript{th} order network $W^{(k)}$ determines the presence or lack of paths of length $k$ of unique edges between any pair of nodes in $W$. For example, the second order network is a multi-graph which has as many edges between a pair of nodes as there are number of paths of length 2 between them in the original network. If the original network is directed, its diffusion paths and its higher order networks will be directed too.
\\
\textbf{Adjacency Matrix:} The $(i,j)$ element contains the number of outgoing stubs from node $i$ to node $j$. In the case of undirected DCSBM \citep{Karrer2011}, diagonal elements will be twice the number self-loops since they correspond to the number of self-loop stubs. However, if the network is directed, the diagonal elements contain the number of self-loops, not twice their value, since self-edges are directed and each has only one outgoing stub.
\\
\textbf{Higher Order Assortativities:} Higher order assortativities measure the extent of unequal diffusion in the network. The k\textsuperscript{th} order assortativity of network $G$ is simply the assortativity of its k\textsuperscript{th} order network $G^{(k)}$.  For example, if we denote the directed adjacency matrix of the second order network as $\mathbf{A}^{(2)}$, then we can define the second order assortativity, $r^{(2)}$, in a manner similar to equations \ref{eq:assort1} and \ref{eq:assort1_helper}. 
\begin{align}
    & r^{(2)} = \frac{\sum \limits_{r} e^{(2)}_{rr} - \sum \limits_{r} a^{(2)}_r b^{(2)}_r}{1 - \sum \limits_{r}a^{(2)}_r b^{(2)}_r}
    \label{eq:assort2}
\end{align}
where the quantity $e^{(2)}_{rs}$ is the fraction of total (directed) paths of length 2 from a node in group $r$ to a node in group $s$, $a^{(2)}_r$ is the fraction of total directed paths of length 2 from a node in group $r$ and $b^{(2)}_r$ is the fraction of total paths of length 2 to a node in group $r$ in the original network.
\begin{align}
    & e^{(2)}_{rs} = \frac{\sum \limits_{i,j} A^{(2)}_{ij} \delta_{g_i,r} \delta_{g_j,s}}{\sum \limits_{i,j} A^{(2)}_{ij}} \quad \quad \quad
    a^{(2)}_r = \sum_s e^{(2)}_{rs} \quad \quad \quad
    b^{(2)}_r = \sum_s e^{(2)}_{sr}
    \label{eq:assort2_helper}
\end{align}
\\
\textbf{Directed Networks:} In what follows we develop our model assuming the network is directed as most social networks do have a notion of direction in edges. The undirected model is very similar to the directed version with the difference that it will replace any pair of parameters that correspond to two incoming and outgoing directions in the directed model with a single parameter.
\\
\textbf{Notation:}
Throughout, we refer to the group a node $i$ belongs to as $g_i$, set of all groups as $G$, the set of all nodes as $N$, total number of edges from group $r$ to $s$ as $m_{rs}$, total out-degree (in-degree) of all nodes in group $r$ as $d_r^o$ ($d_r^i$), total out-degree (in-degree) of node $i$ as $d_i^o$ ($d_i^i$) and the out-degree (in-degree) of node $i$ to group $r$ as $d_{i,r}^o$ ($d_{i,r}^i$).

\subsection{Setup}
Our random graph model is based on the degree-corrected Stochastic Block Model \citep{Karrer2011}. In contrast to DCSBM and instead of correcting for the total degree of each node, we correct for its degree to each group. By correcting for the out-group degree of each node, we can differentiate between networks whose cross-group links are exclusive to a small number of brokers versus those with an equal distribution of cross-group links. We show that by including extra parameters for this correction, the model not only corrects for the degree of each node, but also fits the number of in-group and out-group paths of length 1 and 2 in expectation and as a result the estimated assortativity on paths of length 2 is approximately equal to its observed value.

The main difference with DCSBM and our model is that after conditioning on degree, cross-group links are not distributed equally among all nodes of a group. Instead, each node will have a separate parameter for propensity of linking with each group and the combination of these cross-group propensity parameters determines how cross-group edges are distributed among nodes of a group.
Furthermore, as the network is directed, we introduce one such node-level parameter and one group-level baseline linking parameter for each incoming and outgoing direction. Given these parameters, the number of edges from a node $i$ from group $r$ to a node $j$ from group $s$ is modeled as a Poisson random variable with mean $\theta_{i,s}^{o} \theta_{j,r}^{i} \omega_{rs}$ where  $\theta_{i,s}^{o}$ is the outgoing propensity parameter for node $i$ to group $s$, $\theta_{j,r}^{i}$ is the incoming propensity parameter for node $j$ from group $r$ and $\omega_{rs}$ parameter adjusts the baseline number of edges from group $r$ to $s$. Thus, the expected number of edges from $i$ to $j$ is $\EE{A_{ij}} = \theta_{i,g_j}^{o} \theta_{j, g_i}^{i} \omega_{g_i g_j}$. A nice property of this directed model over undirected DCSBM is that the expected number of self-loops match the expected value of their corresponding diagonal elements without an extra $\frac{1}{2}$ factor since we only count the number of outgoing stubs in the adjacency matrix of a directed network.

We can now express the likelihood function in this model with node-level variation in cross-group linking propensity:
\begin{align}
L(\Theta, \Omega; \mathbf{A}) = & \prod_{i,j} \frac{(\theta_{i,g_j}^{o} \theta_{j, g_i}^{i} \omega_{g_i g_j})^{A_{ij}} }{A_{ij}!} \exp(-\theta_{i,g_j}^{o} \theta_{j, g_i}^{i} \omega_{g_i g_j})
\label{eq:model_likelihood}
\end{align}
where $\Theta$ is the set of node-level outgoing and incoming degree propensity parameters, $\Omega$ is the group-level edge formation propensity parameters, $g_i$ denotes the group of node $i$ and $\mathbf{A}$ is the (directed) adjacency matrix where $A_{ij}$ is the number of outgoing edges from node $i$ to $j$. Given this setup, the MLE for $\Omega$ is as followed:
\begin{align}
    \widehat{\omega}_{rs} &= \frac{m_{rs}}{\sum_{i \in r, j \in s} \widehat{\theta}_{i,s}^o \widehat{\theta}_{j,r}^i}
    \label{eq:omega_mle}
\end{align}
where $m_{rs}$ is the number of outgoing edges from group $r$ to group $s$. The denominator resembles the effective number of pairs for such links. To derive the MLE for $\Theta$, we note that $\theta$ parameters can be arbitrary to within a constant, therefore we must impose additional structure on the model. These constraints can take different forms and one of our contributions is to show that different constraints lead to different models. Below we briefly discuss two constraints and derive their resulting MLE.

\subsection{Node Level Constraint}
One alternative for model structure is to impose a constraint on total propensity of each node, as shown below.
\begin{align}
   \forall i \in N: \quad \sum_{g \in G} \theta_{i,g}^o = 1, \quad \sum_{g \in G} \theta_{i,g}^i = 1
   \label{eq:node_constraint}
\end{align}

This constraint imposes the same fixed value on total propensity of linking to and from all groups for each node. It still allows for cross-group linking variation within each group, as each node can distribute its linking propensity differently. However, the constraint limits the degree variation of all nodes, since the overall linking propensity of each node is fixed.
The MLE of this model for $\Theta$ simplifies to a system of equations, as shown below.
\begin{align}
\begin{split}
    & \forall i \in N \quad \forall g_1, g_2 \in G: \quad \sum_{j \in g_1} (\widehat{\omega}_{g_i g_1} \widehat{\theta}_{j, g_i}^i - \frac{A_{ij}}{\widehat{\theta}_{i, g_1}^o}) = \sum_{j \in g_2} (\widehat{\omega}_{g_i g_2} \widehat{\theta}_{j, g_i}^i - \frac{A_{ij}}{\widehat{\theta}_{i, g_2}^o}) \\
    & \forall i \in N \quad \forall g_1, g_2 \in G: \quad \sum_{j \in g_1} (\widehat{\omega}_{g_1 g_i} \widehat{\theta}_{j, g_i}^o - \frac{A_{ji}}{\widehat{\theta}_{i, g_1}^i}) = \sum_{j \in g_2} (\widehat{\omega}_{g_2 g_i} \widehat{\theta}_{j, g_i}^o - \frac{A_{ji}}{\widehat{\theta}_{i, g_2}^i}) 
    \label{eq:theta_mle}
\end{split}
\end{align}

Combining the MLE equations \ref{eq:omega_mle} and \ref{eq:theta_mle} with the constraint equations \ref{eq:node_constraint}, one can numerically compute the MLE. In general, the maximum likelihood estimates don't have a closed-form solution, but if the observed out-degree and in-degree of all nodes within each group are identical, the MLE takes the following convenient and intuitive form.
\begin{align}
\begin{split}
\textit{if } \quad \forall r \in G \;\; \forall i,j \in r \quad \, d_i^o = d_j^o: \quad \widehat{\theta}_{i,s}^o = \frac{d_{i,s}^o}{d_i^o} \\
\textit{if } \quad \forall r \in G \;\; \forall i,j \in r \quad \, d_i^i = d_j^i: \quad \widehat{\theta}_{i,s}^i = \frac{d_{i,s}^i}{d_i^i}
\end{split}
\end{align}
This result implies that the propensity of linking to a group $s$ is simply the observed fraction of the node's total degree to that group.

\subsection{Group Level Constraint}
Another alternative for model structure is to impose a constraint on total propensity of all nodes within a group, as shown below. This model will be the main focus of our work and has close resemblance to DCSBM but with extra desirable properties.
\begin{align}
    & \forall r,s \in G: \quad \sum_{i \in r} \theta_{i,s}^o = 1, \quad \sum_{i \in r} \theta_{i,s}^i = 1
\end{align}
The constraint states that the total propensity of linking to and from group $s$ is fixed among all nodes of group $r$. Variation in cross-group linking among nodes of a group can still exist. Naturally, a good model will distribute the propensity supply of each group according to cross-group degree of the nodes within that group. The MLE of the model for $\Theta$ simplifies to the following intuitive forms:
\begin{align}
\begin{split}
    \forall r,s \in G \quad \forall i \in r: \quad \widehat{\theta}_{i,s}^o &= \frac{d_{i,s}^o}{m_{rs}} \\
    \forall r,s \in G \quad \forall i \in r: \quad \widehat{\theta}_{i,s}^i &= \frac{d_{i,s}^i}{m_{rs}}
\end{split}
\end{align}
In contrast to the previous constraint at the node-level which led to within-node fractions, the MLE for linking propensity to a group $s$ with the group-level constraint becomes the within-group fraction: the observed fraction of total cross-group degree that originates from the focal node. This estimate closely resembles that of the propensity parameter in regular DC-BM with the exception that MLE fractions in DCSBM did not differentiate between the degrees to each group. Given the estimates above for propensity parameters, the MLE for group-level parameters from equation \ref{eq:omega_mle} simplifies to the number of cross-group edges: 
\begin{align}
    \widehat{\omega}_{rs} &= m_{rs} \label{eq:omega_mle_simplified}
\end{align}

\subsubsection{Frequency of Diffusion Paths Under MLE Model}
\label{sec:assort_bias}
Before deriving the expected number of cross-group edges from the model fit, we compute a few useful parameters that result from the fitted model: the expected number of edges between any two nodes and the expected out-degree (in-degree) of a node to a group. The variables with a hat are generated by the model and refer to the corresponding observed quantity with same symbol.
\begin{align}
    \EE{\widehat{A}_{ij}} &= \frac{d_{i,g_j}^o d_{j,g_i}^i}{m_{g_i g_j}} \label{eq:expected_adjacency} \\
    \EE{\widehat{d}_{i,s}^o} &= \EE{\sum_{j \in s} \widehat{A}_{ij}} = d_{i,s}^o   \label{eq:expected_outdeg} \\
    \EE{\widehat{d}_{i,s}^i} &= \EE{\sum_{j \in s} \widehat{A}_{ji}} = d_{i,s}^i \label{eq:expected_indeg}
\end{align}

We now show that the fitted model matches not only the observed number of paths of length 1 but also the observed number of paths of length 2 between any two groups in expectation, even though the model does not explicitly account for it. Throughout, we assume that traversing the same edge twice is not permissible (e.g. paths cannot use a self-loop twice). However, traversing from a node to its neighbor and back to itself is allowed as long as there is a directed edge in each direction. This is possible under our analysis since edges with different directions between any pair are drawn independently and considered different.

As the first step, we show that that expected number of edges between any two groups in the MLE model matches that of the observed network. Below, we denote the observed and (random) model-generated number of paths of length $k$ from group $r$ to group $s$ by $P_{rs}^{(k)}$ and $\widehat{P}_{rs}^{(k)}$ respectively.
\begin{align}
\begin{split}
\EE{\widehat{P}_{rs}^{(1)}} &= \sum_{i \in r} \sum_{j \in s} \frac{d_{i,s}^o d_{j,r}^i}{m_{rs}} = m_{rs} \\
&= P_{rs}^{(1)}
\end{split}
\label{eq:expected-pl1}
\end{align}
We used equation (\ref{eq:expected_adjacency}) in the first line above.  We now show a similar result for paths of length 2. First, we show that the expected number of paths of length 2 between two different groups, $r\neq s$, matches the observed network.
\begin{align}
\begin{split}
\EE{\widehat{P}_{rs}^{(2)}} 
&= \sum_{j} \sum_{i \in r, k \in s} \EE{\widehat{A}_{ij}} \EE{\widehat{A}_{jk}} \\
&= \sum_{j} \EE{\widehat{d}_{j,r}^{i}} \EE{\widehat{d}_{j,s}^{o}} \\
&= \sum_{j} d_{j,r}^{i} d_{j,s}^{o} \\
&= P_{rs}^{(2)}
\end{split}
\label{eq:diff_group_pl2}
\end{align}
In the first line above, we used the fact that edges are independent and in the third line, we relied on equations (\ref{eq:expected_outdeg}) and (\ref{eq:expected_indeg}). We now show that the expected number of paths of length 2 within a single group is also the same as the observed value. For this result to hold, we must assume the observed networks does not have self-loops. In the appendix, we characterize the bias on the number of in-group paths of length 2 if the observed network has self-loops and show that it vanishes compared to the total number of paths as network size grows.  

\begin{align}
\begin{split}
\EE{\widehat{P}_{rr}^{(2)}} &= \sum_{j} \sum_{\substack{i,k \in r \\ i \neq k}} \EE{\widehat{A}_{ij}} \EE{\widehat{A}_{jk}} + \sum_{j} \sum_{\substack{i \in r \\ i \neq j}} \EE{\widehat{A}_{ij}} \EE{\widehat{A}_{ji}} + \sum_{j \in r} \EE{\widehat{A}_{jj}(\widehat{A}_{jj} - 1)} \\
&= \sum_{j} \sum_{\substack{i,k \in r \\ i \neq k}} \EE{\widehat{A}_{ij}} \EE{\widehat{A}_{jk}} + \sum_{j} \sum_{\substack{i \in r \\ i \neq j}} \EE{\widehat{A}_{ij}} \EE{\widehat{A}_{ji}} + \sum_{j \in r} \EE{\widehat{A}_{jj}}^2 \\
&= \sum_{j} \sum_{\substack{i,k \in r}} \EE{\widehat{A}_{ij}} \EE{\widehat{A}_{jk}} \\
&= \sum_{j} \EE{\widehat{d}_{j,r}^{i}} \EE{\widehat{d}_{j,r}^{o}} \\
&= \sum_{j} d_{j,r}^{i} d_{j,r}^{o} \\
&= P_{rr}^{(2)}
\end{split}
\label{eq:same_group_pl2}
\end{align}
The first line uses the fact that traversing the same edge, including a self-loop in the last term, is not allowed twice and since the network is directed, edges in different directions between the same pair of nodes are considered different and independent, $\widehat{A}_{ij} \perp \widehat{A}_{ji} $. The last term in the second line relies on the fact that the number of edges has a Poisson distribution.
The last line uses our assumption that the observed network does not have any self-loops.
The appendix shows that if the observed network has self-loops, then the estimated number of within-group paths of length 2 would be biased positively by the total number of self-loops within the group.

Finally we note that all the unbiasedness results above on the number of paths of length 1 and length 2 also hold if the network is undirected and traversing the same undirected edge is not allowed twice. In an undirected network, $\widehat{A}_{ij}$ and $\widehat{A}_{ji}$ are assumed to be the same edge hence the second term in the first line of equation (\ref{eq:same_group_pl2}) becomes $\EE{\widehat{A}_{ij}(\widehat{A}_{ij}-1)}$. Without this assumption, there will be a bias in the number of in-group paths of length 2 ($\widehat{P}_{rr}^{(2)}$) roughly equal to the total number of edges adjacent to that group which again vanishes compared to the total number of in-group paths as the size of network grows.

\subsubsection{Asymptotic Behavior of Diffusion Assortativity Under MLE Model}
\label{sec:assort_consistency}
First, we quickly prove a simple extension of weak law of large numbers which we will use in our proof of diffusion assortativity consistency.

\begin{lemma} Let ${\{X_i\}}_1^\infty$ be a sequence of independent random variables with $E[X_i] = \mu_i$ and $Var(X_i) = \sigma_i^2$. If the sequence of variances ${\{\sigma_i^2\}}_1^\infty$ is bounded, then $\frac{\sum_i^n X_i}{n} \rightarrow \frac{\sum_i^n \mu_i}{n}$ in probability.
\label{prop:weak_law}
\end{lemma}

\begin{proof}
Let $S_n = \frac{\sum_i^n X_i}{n}$ and $\mu = \frac{\sum_i^n \mu_i}{n}$, then $Var(S_n) = \frac{\sum_i^n \sigma_i^2}{n^2} \rightarrow 0$. This follows from simple application of Chebychev's inequality.
\begin{align*}
    P(|S_n - \mu| \ge \epsilon) \le \frac{Var(S_n)}{\epsilon^2} \rightarrow 0
\end{align*}
\end{proof}

\begin{proposition} Let $n_r$ be the size of nodes in group $r$ and $n = \sum_r n_r$ be the size of all nodes in the network. If $\hat{e}_{rs}^{(2)}$ is determined from the sampled network according to equation (\ref{eq:assort2_helper}) and $\mathbf{A} \neq 0$, then $\hat{e}_{rs}^{(2)} \xrightarrow{p} e_{rs}^{(2)}$ as $n \rightarrow \infty$.
\label{prop:fractions}
\end{proposition}
\begin{proof} 
Below we denote the adjacency matrix of the second order network of the sampled network as ${\hat{\mathbf{A}}}^{(2)}$. We allow for traversing the same edge multiple times and show that prohibiting them does not affect the result.
\begin{align*}
    \hat{e}^{(2)}_{rs} &= \frac{\sum \limits_{i \in r,j \in s} \hat{A}^{(2)}_{ij}}{\sum \limits_{i,j} \hat{A}^{(2)}_{ij}} \\
    &= \frac{\sum \limits_{i \in r,j \in s,k} \hat{A}_{ik} \hat{A}_{kj}}{\sum \limits_{i,j,k} \hat{A}_{ik} \hat{A}_{kj}} \\
    &= \frac{n_r n_s}{n^2}
     \frac{\frac{\sum \limits_{i \in r,j \in s,k} \hat{A}_{ik} \hat{A}_{kj}}{n_r n_s n}}
          {\frac{\sum \limits_{i,j,k} \hat{A}_{ik} \hat{A}_{kj}}{n^3}}
\numberthis \label{eq:prop2-intermediate1} 
\end{align*}
In the second line above, we allowed for traversing the same edge twice. This can happen only if $i=j=k$ (self-loops). The terms $\hat{A}_{ik}$ and $\hat{A}_{kj}$ are two Poisson random variables with finite mean and variance, thus their product also has finite mean and variance. By applying lemma \ref{prop:weak_law}, we get 
\begin{align*}
\frac{\sum \limits_{i \in r,j \in s,k} \hat{A}_{ik} \hat{A}_{kj}}{n_r n_s n} \xrightarrow{p} \frac{\sum \limits_{i \in r,j \in s,k} E[\hat{A}_{ik} \hat{A}_{kj}]}{n_r n_s n}
\numberthis \label{eq:prop2-intermediate2} 
\end{align*}
The terms $\hat{A}_{ik}$ and $\hat{A}_{kj}$ are independent unless $i = j = k$ which can only happen if $r=s$. Below we assume this is the case but the results remain the same even if $r \neq s$.
\begin{align*}
\lim_{n \to \infty} \frac{\sum \limits_{i \in r,j \in s,k} E[\hat{A}_{ik} \hat{A}_{kj}]}{n_r n_s n} &= 
\lim_{n \to \infty} \frac{\sum \limits_{i \in r,j \in s,k} E[\hat{A}_{ik}] E[\hat{A}_{kj}]}{n_r n_s n} +
\lim_{n \to \infty} \frac{\sum \limits_{i \in r} E[\hat{A}_{ii}^2]}{n_r n_s n} 
\numberthis \label{eq:prop2-intermediate3}  \\
&= \lim_{n \to \infty} \frac{\sum \limits_{i \in r,j \in s} d_i^o d_j^i}{n_r n_s n}
\end{align*}
In the second line we used the fact that variance of self-loops is finite. Combining the result above with equation \ref{eq:prop2-intermediate2} and performing the same analysis for the denominator in equation \ref{eq:prop2-intermediate1}, we get the following convergences.
\begin{align}
\begin{split}
\frac{\sum \limits_{i \in r,j \in s,k} \hat{A}_{ik} \hat{A}_{kj}}{n_r n_s n} &\xrightarrow{p} \frac{\sum \limits_{i \in r,j \in s} d_i^o d_j^i}{n_r n_s n} \\
\frac{\sum \limits_{i,j,k} \hat{A}_{ik} \hat{A}_{kj}}{n^3} &\xrightarrow{p} \frac{\sum \limits_{i,j} d_i^o d_j^i}{n^3}
\end{split}
\label{eq:prop2-intermediate4} 
\end{align}
Combining equations \ref{eq:prop2-intermediate1} and \ref{eq:prop2-intermediate4} and the fact that $\sum \limits_{i,j} d_i^o d_j^i \neq 0$, we get the result using the continuous mapping theorem:
\begin{align}
    \hat{e}^{(2)}_{rs} \xrightarrow{p} e^{(2)}_{rs} 
\end{align}
\end{proof}

\noindent \textit{Remark 1}. If we had not allowed for traversing the same edge multiple times, the second term in equation \ref{eq:prop2-intermediate3} would not be present and the final limit would be the same.

\noindent \textit{Remark 2}. In case of an undirected network, we would have the same convergence results as long as $n_s \rightarrow \infty$ for all $s$ when $n \rightarrow \infty$. In this case, the second term in equation \ref{eq:prop2-intermediate3} would be replaced by $\frac{\sum \limits_{i \in r,k} E[\hat{A}_{ik}^2]}{n_r n_s n}$ which still tends to zero as $n \rightarrow \infty$.

\begin{proposition}
The sampled assortativity on paths of length 2 from the MLE model converges in probability to the observed assortativity on paths of length 2.
\label{prop:assort_consistency}
\end{proposition}
\begin{proof}
The assortativity on paths of length 2 from a sampled network is defined as below.
\begin{align*}
    \hat{r}^{(2)} = \frac{\sum \limits_{r} \hat{e}_{rr}^{(2)} - \sum \limits_{r} \hat{a}_r^{(2)} \hat{b}_r^{(2)}}{1 - \sum \limits_{r}\hat{a}_r^{(2)} \hat{b}_r^{(2)}}
    \quad \quad \quad
    \hat{a}^{(2)}_r = \sum_s \hat{e}^{(2)}_{rs} \quad \quad \quad
    \hat{b}^{(2)}_r = \sum_s \hat{e}^{(2)}_{sr}
    \label{eq:prop3-intermediate1}
\end{align*}
where all quantities $\hat{e}_{rs}^{(2)}, \hat{a}_r^{(2)}, \hat{b}_r^{(2)}$ are determined from the sampled network. The result follows using proposition \ref{prop:fractions} on each individual term of $\hat{r}^{(2)}$ and the continuous mapping theorem. In the application of continuous mapping theorem we rely on $\sum \limits_{r} a_r^{(2)} b_r^{(2)} < 1$ since $\sum \limits_{r,s} e_{rs}^{(2)} = 1$.
\end{proof}

\noindent \textit{Remark}. In case of an undirected network, the same result holds as long as $n_s \rightarrow \infty$ for all $s$ when $n \rightarrow \infty$.

\section{Empirical Results}
\label{sec:results}
In this section, we show the same results as in section \ref{sec:dcsbm-empirical} but using our (directed) model instead of the (directed) DCSBM. In particular, we consider the same networks as before and compare their observed assortativity on paths of length 2 and 3 versus the distribution of those quantities generated by the MLE model. As shown above, we would expect the observed assortativity on paths of length 2 to be close to the predicted value by the fitted model, since the networks are large enough. Even if the networks are not large enough for proposition \ref{prop:assort_consistency} to be valid, the bias in model-generated assortativity should be small since the number of in-group and out-group paths of length 2 from the model match the corresponding observed values in expectation, as shown in section \ref{sec:assort_bias}.

\begin{figure}
     \centering
     \begin{subfigure}[b]{\textwidth}
         \centering
         \includegraphics[width=0.47\textwidth]{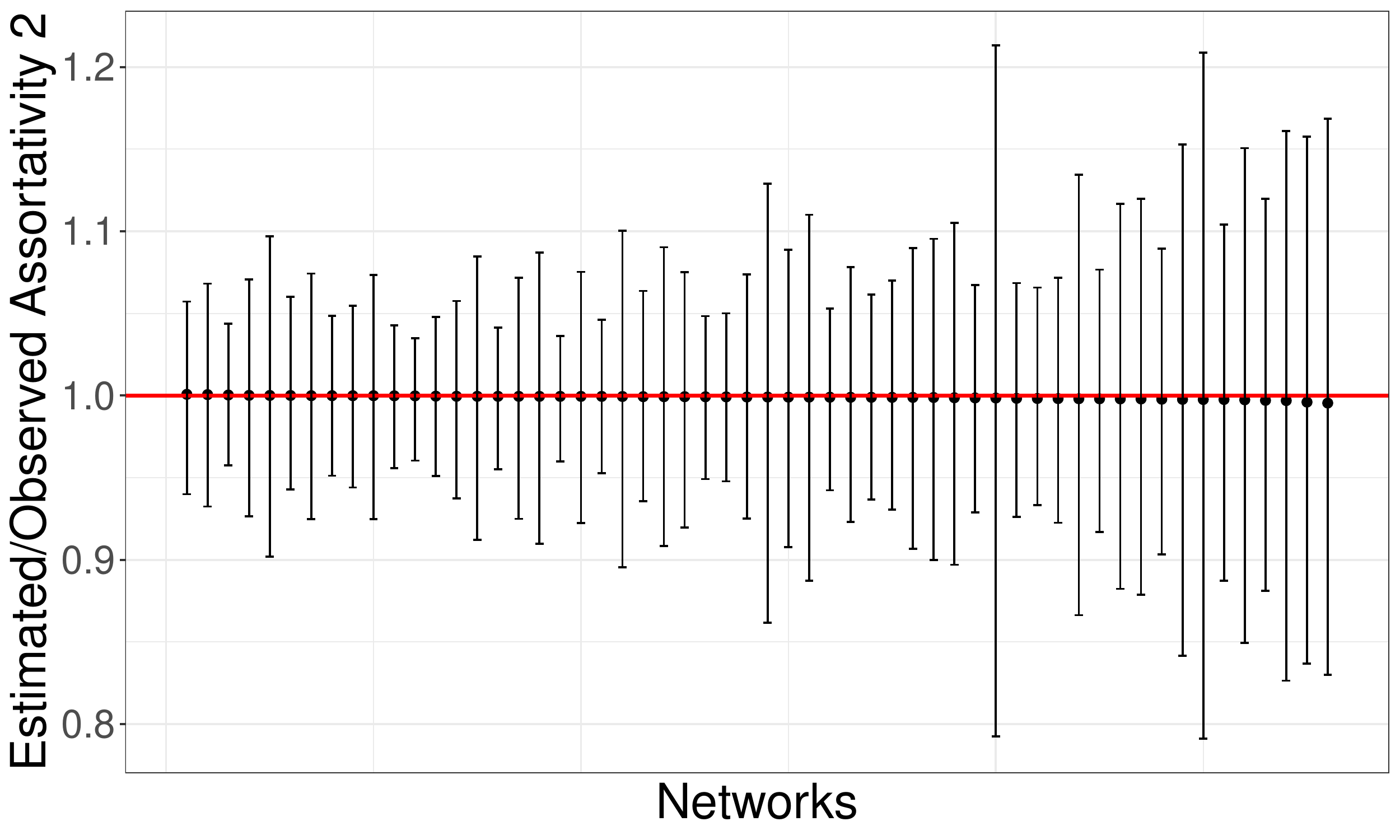}
         \hspace{0.04\textwidth}
         \includegraphics[width=0.47\textwidth]{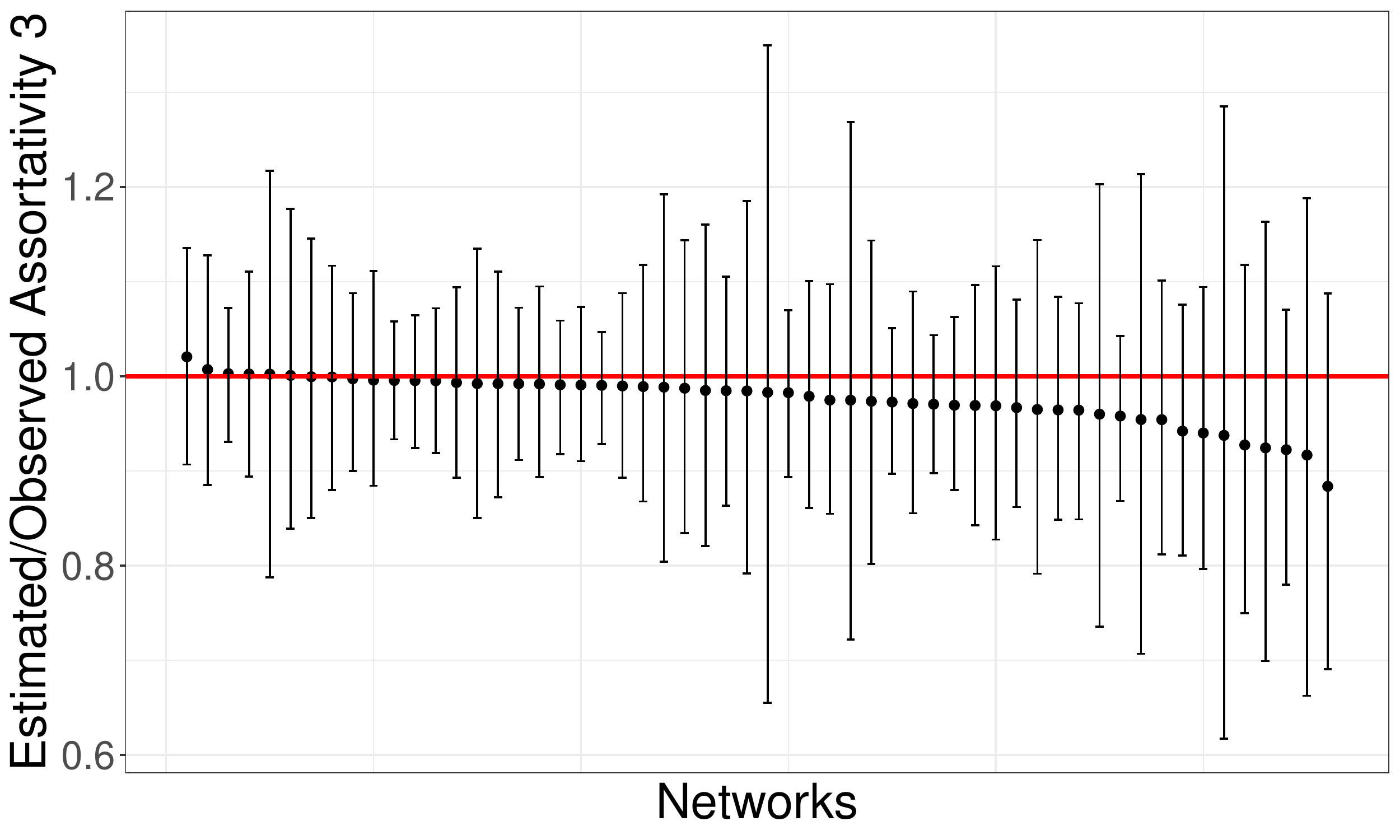}
         \caption{Gender}
     \end{subfigure}
     \par\bigskip 
     \begin{subfigure}[b]{\textwidth}
         \centering
         \includegraphics[width=0.47\textwidth]{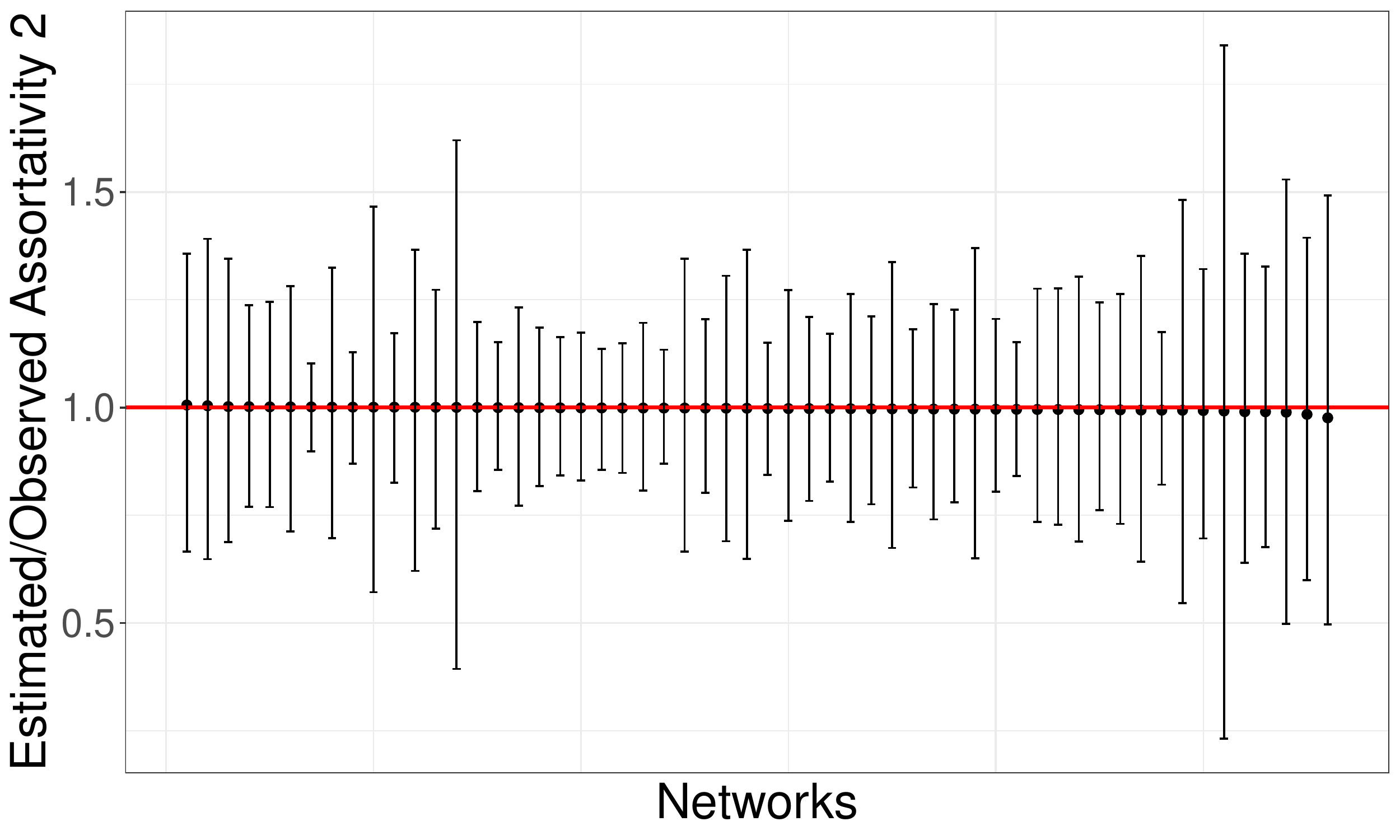}
         \hspace{0.04\textwidth}
         \includegraphics[width=0.47\textwidth]{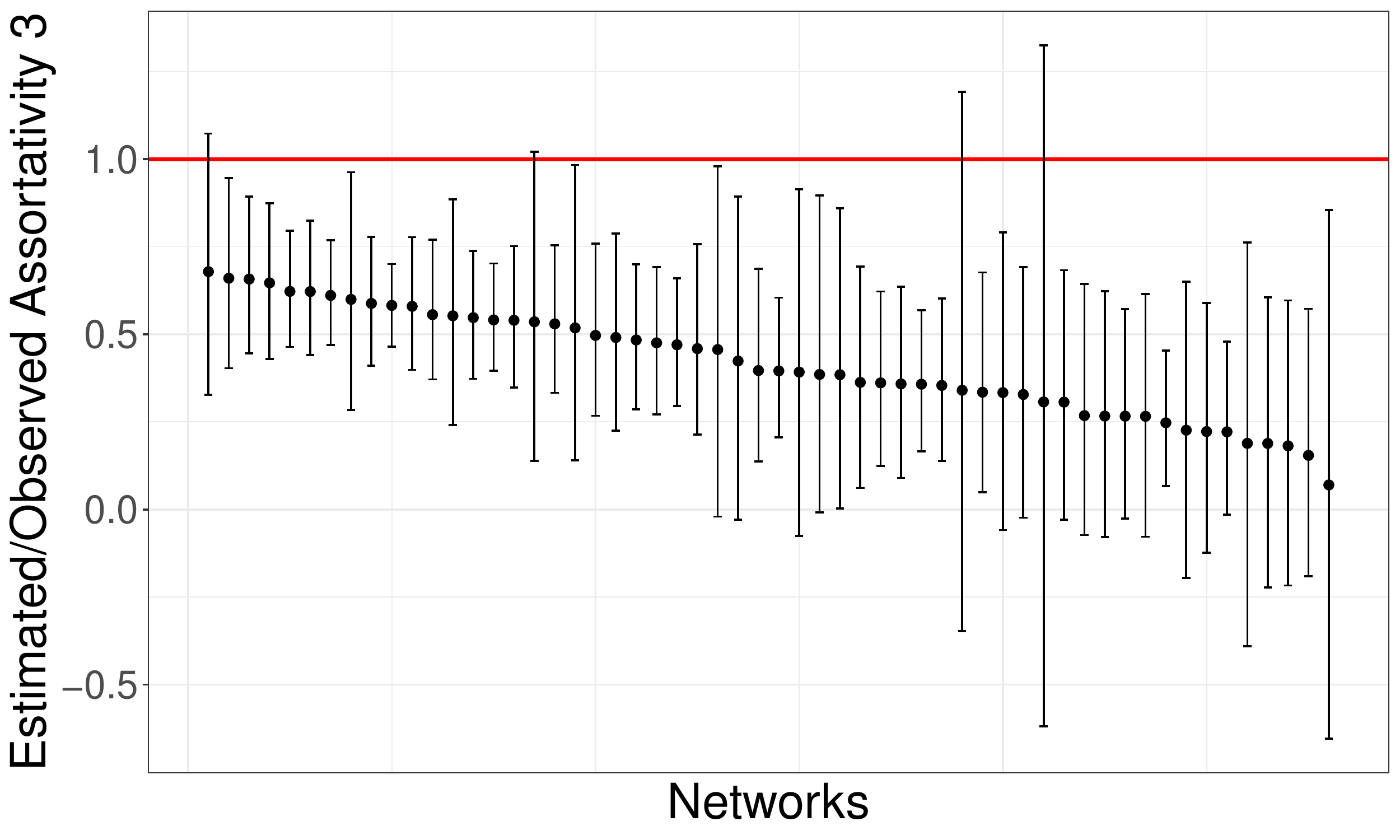}
         \caption{Race}
     \end{subfigure}
    \caption{The distribution of predicted over observed ratio of assortativities on paths of length 2 (left column), $\frac{\hat{r}^{(2)}}{r^{(2)}}$, and 3 (right column), $\frac{\hat{r}^{(3)}}{r^{(3)}}$, from our model along gender (top row) and racial (bottom row) groups. Bars correspond to 95\% confidence interval and each bar corresponds to one school network. Networks are sorted in descending order of the point estimate.}
    \label{fig:mixed_dcsbm_observed_vs_predicted_assorts}
\end{figure}

Figure \ref{fig:mixed_dcsbm_observed_vs_predicted_assorts} compares the distribution of assortativites generated by the model against the observed values along gender and racial groups. The figure is produced exactly as figure \ref{fig:dcsbm_observed_vs_predicted_assorts}, with the same networks and attributes, except that the fitted model accounts for brokerage. First, we observe that in contrast to regular DCSBM, our model accurately captures assortativity along paths of length 2 for both attributes. This is expected since our model fitted through MLE matches the observed frequency of paths of length 2 in expectation. Second, Even though our model does not make any guarantees about assortativity on longer paths, it nevertheless provides a close match with observed assorativity on paths of length 3, at least along gender groups. The observed gender assortativity on paths of length 3 is not significantly different from the generated distribution by the model in any of the networks. However, the model consistently underestimates racial assortativity along paths of length 3. This is mainly because the the absolute level of assortativity along race is much smaller than gender, with values that are often close to zero (only 12 out 56 networks have racial assortativity greater than 0.1). At such small values the model requires higher precision and small absolute differences can make its predictions significantly different relative to the observations. Nevertheless, comparing the distribution of generated racial assortativities along paths of length 3 in figure \ref{fig:mixed_dcsbm_observed_vs_predicted_assorts} with figure \ref{fig:dcsbm_observed_vs_predicted_assorts}, we observe that our model's predictions are at least an order of magnitude closer to observations than DCSBM.

\section{Model Selection}
\label{sec:model_selection}

When inter-group linking propensities are homogenous within each group, i.e. little or no brokerage, DCSBM is a better model than our model since our model will lead overfitting since it has more parameters. However, substantial level of brokerage in cross-group linking justifies the use of our model over DCSBM. In such situations, model selection allows us to pick the right model. Log likelihood ratio tests whether our model captures salient patterns in the networks, in ways that lead to statistically significant improvements in its goodness of fit over DCSBM as the null model.
In deriving the log-likelihood ratio, we note that DCSBM is nested inside our model since it is obtained by imposing a homogeneity constraint on cross-group linking propensities of our model leading to a single out-degree, $\theta_j^o$, and a single in-degree, $\theta_j^i$, parameter for each node:
\begin{align}
    & \forall j \in N, \; \forall r,s \in G: \quad \theta_{j,r}^o = \theta_{j,s}^o = \theta_j^o, \quad  \theta_{j,r}^i = \theta_{j,s}^i = \theta_j^i
\end{align}

The log-likelihood ratio statistic comparing our model to regular DCSBM can be expressed as:
\begin{align}
    \hat{\lambda} = \log{\frac{\sup_{(\Theta, \Omega) \in \mathcal{P} } L(\Theta, \Omega; \mathbf{A})}{\sup_{(\Theta, \Omega) \in \mathcal{P}_0} L(\Theta, \Omega; \mathbf{A})}}
\end{align}
where $\mathcal{P}_0$ and $\mathcal{P}$ denote the restricted and full model parameter spaces respectively.
Large values of $\hat{\lambda}$ test statistic indicates support for the full model that it provides statistically significant improvements over DCSBM.

Since DCSBM as the null is a special case of our model as the alternative, one could appeal to the Wilks theorem \citep{bickel2015mathematical} which states that the test statistic $-2\log(\hat{\lambda})$ is asymptotically distributed as chi-squared with the number of constraints that we must impose on our model to obtain DCSBM as its degrees of freedom. This type of hypothesis testing that uses approximate likelihood ratio chi-squared statistic for network models has been used before \citep{wang1987}. However, the classical results on $\chi^2$ distribution is not valid in this case, as the number of parameters in both the null and alternative increase with $n$. More importantly, the difference in dimensionality of null and alternative is $|G|(|N|-1)$ since our model now contains a degree parameter to each group as opposed to a single degree parameter in the null. Wilks theorem is not valid in this scenario, since the difference in dimensionality increases with the sample size, a point made much earlier in \citep{fienberg1981} about the growing number of parameters in $p_1$ network models.

The null distribution of goodness of fit tests in similar applications with growing number of parameters has been developed in the literature, as was also hinted at \citep{fienberg1981}. A related problem to the analysis of the LLR here is the development of goodness-of-fit tests for large multinomials when the number of cells increases with the sample size. It was shown that the log-likelihood ratio (LLR) statistic for such a scenario follows a normal distribution whose mean and variance is unrelated to the classical chi-squared prediction \citep{zelterman1987goodness,koehler1980empirical}. More recently, asymptotic normality of LLR with maximum likelihood and variational approximations was shown in the context of simple SBM without degree corrections \citep{bickel2013asymptotic}. A recent analysis investigated the null distribution of a very relevant LLR statistic to our development here \citep{yan2014model}. This work addresses the issue of model selection between regular SBM and DCSBM and establishes the asymptotic normality of LLR whose mean and variance depend on the sparsity of the network. The issues encountered in our analysis and \citep{yan2014model} are similar; however as opposed to the development in \citep{yan2014model} where the LLR null model is the SBM, the null model we are testing against is DCSBM. This makes the derivation of asymptotic distribution more challenging since the number of parameters in our null model grows with the sample size whereas the number of parameters in SBM is fixed. Nevertheless, our development of the LLR and its asymptotic distribution was very much influenced by this recent work \citep{yan2014model}. Similar to their work, we establish the asymptotic normality of LLR and show it has a slightly larger mean if the network is sparse.

\subsection{Asymptotic Normality of the Log-Likelihood Ratio}
We start by deriving the expression for the LLR statistic, as the log ratio between the maximum likelihood estimates from our model and estimates from DCSBM.
\begin{align*}
    \hat{\lambda} &=
    \log{
        \frac
        {\prod_{i,j \in N} (\widehat{\theta}_{i,g_j}^{o} \widehat{\theta}_{j, g_i}^{i} \widehat{\omega}_{g_i g_j})^{A_{ij}} \exp(-\widehat{\theta}_{i,g_j}^{o} \widehat{\theta}_{j, g_i}^{i} \widehat{\omega}_{g_i g_j})}
        {\prod_{i,j \in N} (\widehat{\theta}_{i}^{o} \widehat{\theta}_{j}^{i} \widehat{\omega}_{g_i g_j})^{A_{ij}} \exp(-\widehat{\theta}_{i}^{o} \widehat{\theta}_{j}^{i} \widehat{\omega}_{g_i g_j})}
    }\\
    &= \sum_{i,j \in N} \Big[ A_{ij} ( \log{ \frac{d_{i,g_j}^o}{m_{g_i g_j}} } + \log{\frac{d_{j,g_i}^i}{m_{g_i g_j}} } - \log{m_{g_i g_j}} ) - \frac{d_{i,g_j}^o}{m_{g_i g_j}} \frac{d_{j,g_i}^i}{m_{g_i g_j}} m_{g_i g_j} \Big] \\
    &- \sum_{i,j \in N} \Big[ A_{ij} ( \log{ \frac{d_{i}^o}{d_{g_i}^o}} + \log{\frac{d_{j}^i}{d_{g_j}^i}} - \log{m_{g_i g_j}} ) - \frac{d_{i}^o}{d_{g_i}^o} \frac{d_{j}^i}{d_{g_j}^i} m_{g_i g_j} \Big] \\
    &= \sum_{i,j \in N} \Big[ d_{i,g_j}^o \log{d_{i,g_j}^o} +  d_{j,g_i}^i \log{d_{j,g_i}^i} - d_{i,g_j}^o \log{m_{rs}} -  d_{j,g_i}^i \log{m_{rs}} \Big] \\
    &- \sum_{i,j \in N} \Big[ d_{i}^o \log{d_{i}^o} +  d_{j}^i \log{d_{j}^i} - d_{i}^o \log{d_{g_i}^o} -  d_{j}^i \log{d_{g_j}^i} \Big]
\end{align*}
where we plugged in the maximum likelihood estimates for both our model and the DCSBM in the second line. The log-likelihood ratio simplifies to the following form involving individual node and group degrees. The quantities below were all defined in section \ref{sec:model_prelim}.
\begin{align}
\begin{split}
    \hat{\lambda} &= \sum_{\substack{i \in N \\ g \in G}} \Big[ d_{i,g}^o \log{d_{i,g}^o} +  d_{i,g}^i \log{d_{i,g}^i} \Big] - \sum_{i \in N} \Big[ d_{i}^o \log{d_{i}^o} +  d_{i}^i \log{d_{i}^i} \Big] \\
    &- \sum_{r,s \in G} 2m_{rs} \log{m_{rs}}  + \sum_{r \in G}\Big[ d_r^o \log{d_r^o} + d_r^i \log{d_r^i} \Big]
\end{split}
\label{eq:llr}
\end{align}

Under the null and assuming DCSBM is the true model, then each term in equation (\ref{eq:llr}) will have a Poisson distribution whose parameters depend on the true DCSBM model:
\begin{align}
\begin{aligned}
& d_{i,g}^o \sim \text{Poisson}(\theta_i^o \omega_{g_i g})
&& d_{i,g}^i \sim \text{Poisson}(\theta_i^i \omega_{g g_i})  \\
& d_{i}^o \sim \text{Poisson}(\theta_i^o \textstyle\sum_{g \in G}\omega_{g_i g}) 
&& d_{i}^i \sim \text{Poisson}(\theta_i^i \textstyle\sum_{g \in G} \omega_{g g_i}) \\
& d_r^o \sim \text{Poisson}(\textstyle\sum_{s \in G}\omega_{rs}) && d_r^i \sim \text{Poisson}(\textstyle\sum_{s \in G}\omega_{sr}) &&& m_{rs} \sim \text{Poisson}(\omega_{rs})
\end{aligned}
\label{eq:poisson_expectations}
\end{align}
Each term in equation (\ref{eq:llr}) is an independent Poisson random variable. As long as a weak notion of non-sparsity holds, namely that the expected (in-)out-degree of each node to each group, $\theta_i^o \omega_{g_i g}$ and $\theta_i^i \omega_{g g_i}$, does not shrink as network grows, then Lindeberg central limit theorem holds and $\hat{\lambda}$ will approach a normal distribution as $n \to \infty$. This is justified since as the number of nodes increase, then individual node and group degree terms become independent.

\subsection{Expectation of the Log-Likelihood Ratio}
Using equation (\ref{eq:llr}) and the expected value of its terms in equation (\ref{eq:poisson_expectations}), we can write the expected value of the LLR under the null that DCSBM is the true model:
\begin{align}
\begin{split}
    \EE{\hat{\lambda}} &= \sum_{\substack{i \in N \\ g \in G}} \Big[ f(\theta_i^o \omega_{g_i g}) +  f(\theta_i^i \omega_{g g_i}) \Big] -
    \sum_{i \in N} \Big[ f(\theta_i^o \textstyle\sum_{g \in G} \omega_{g_i g}) +  f(\theta_i^i \textstyle\sum_{g \in G} \omega_{g g_i}) \Big] \\
    &- \sum_{r,s \in G} 2f(\omega_{rs})  + \sum_{r \in G}\Big[ f(\textstyle\sum_{s \in G} \omega_{rs}) + f(\textstyle\sum_{s \in G} \omega_{sr}) \Big]
\end{split}
\label{eq:llr_exp}
\end{align}
where we have defined $f(\mu) = \EE{X\log{X}}$ for $X \sim \text{Poisson}(\mu)$.

\begin{theorem} If the following conditions holds,
\begin{align*}
    & \forall i \in N, \forall g \in G \quad
    \theta_i^o \omega_{g_i g} \to \infty
    \quad \text{as} \quad n \to \infty \\
    & \forall i \in N, \forall g \in G \quad
    \theta_i^i \omega_{g g_i} \to \infty
    \quad \text{as} \quad n \to \infty
\end{align*}
then, $\EE{\hat{\lambda}} = (|G| - 1)(|N| - |G|)$
\label{theorem:llr_exp}
\end{theorem}
\begin{proof} 
Taylor series expansion of $f(\mu)$ around $\mu$ becomes:
\begin{align}
    f(\mu) = \mu \log{\mu} + \frac{1}{2} + \frac{1}{12\mu} + \frac{1}{12\mu^2} + O(\frac{1}{\mu^3})
    \label{eq:taylor_series}
\end{align}
The condition implies we can only keep the first two terms from the Taylor series expansion of all quantities in equation (\ref{eq:llr_exp}).
\begin{align*}
    \EE{\hat{\lambda}} &= \sum_{\substack{i \in N \\ g \in G}} \Big[ \theta_i^o \omega_{g_i g} \log{(\theta_i^o \omega_{g_i g})} +  \theta_i^i \omega_{g g_i} \log{(\theta_i^i \omega_{g g_i})} + 1\Big] \\
    &- \sum_{i \in N} \Big[ \theta_i^o \textstyle\sum_{g \in G} \omega_{g_i g} \log{(\theta_i^o \textstyle\sum_{g \in G} \omega_{g_i g}) }  +  \theta_i^i \textstyle\sum_{g \in G} \omega_{g g_i} \log{(\theta_i^i \textstyle\sum_{g \in G} \omega_{g g_i}) }  + 1 \Big] \\
    &- \sum_{r,s \in G} \Big[2\omega_{rs} \log{\omega_{rs}} + 1 \Big] \\
    &+ \sum_{r \in G}\Big[ \textstyle\sum_{s \in G} \omega_{rs} \log{( \textstyle\sum_{s \in G} \omega_{rs})}  + \textstyle\sum_{s \in G} \omega_{sr} \log{( \textstyle\sum_{s \in G} \omega_{sr})}  \Big]
\end{align*}
The equation above simplifies to $ (|G| - 1)(|N| - |G|)$ by using the constraints 
\begin{align*}
\forall r,s \in G: \quad \sum_{i \in r} \theta_{i,s}^o = 1, \quad \sum_{i \in r} \theta_{i,s}^i = 1
\end{align*}
\end{proof} 
Theorem \ref{theorem:llr_exp} states that in the limit of dense networks, the expected value of LLR essentially matches the value predicted by Wilks theorem since $ (|G| - 1)(|N| - |G|)$ is in fact the number of constraints one can put on parameters of the model to recover DCSBM.

As it relates to the distribution of LLR and according to the conditions in theorem \ref{theorem:llr_exp}, a network is considered to be sparse if:
\begin{align*}
\exists i \in N \;\, g \in G: \quad\; \theta_i^o \omega_{g_i g}=O(1) \quad or \quad \theta_i^i \omega_{g g_1} = O(1)
\end{align*}
The expected value of LLR in a sparse network will be larger than a corresponding dense network, suggesting that in the case of sparse networks the risk of overfitting and rejecting a true DCSBM is higher. Nevertheless, one can obtain an accurate value for expected value of LLR by using the the Taylor series expansion of $f(\mu)$ in equation (\ref{eq:taylor_series}) including its higher order terms and conducting the sum in equation (\ref{eq:llr_exp}) numerically.

\begin{figure}[t]
\centering
\includegraphics[width=0.97\textwidth]{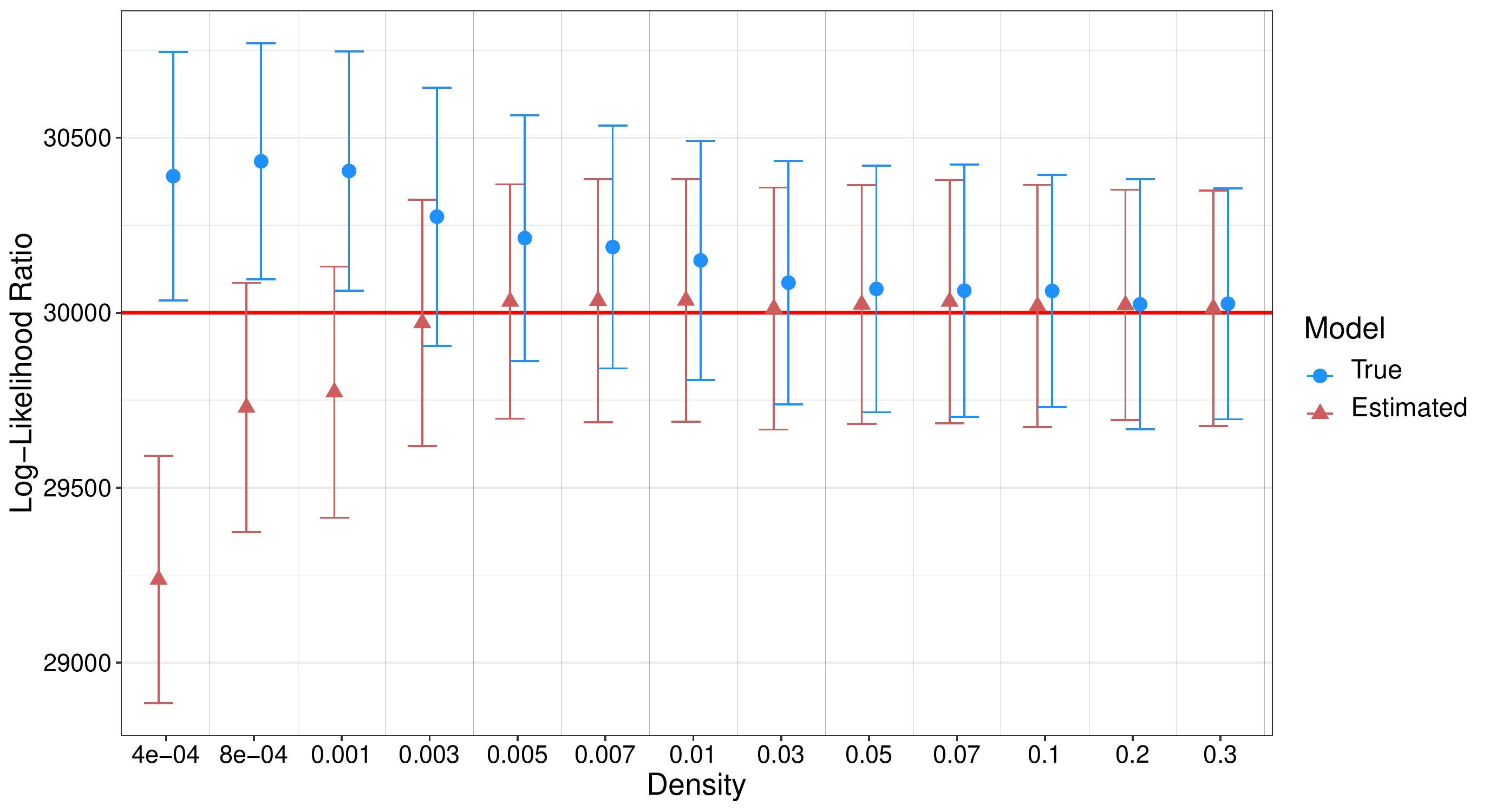}
\caption{The distribution of LLR in DCSBM generated networks with 30,000 nodes and varying density. True (estimated) model distributions are generated by sampling from the true (estimated) model. Bars correspond to two standard errors.}
\label{fig:analytical_vs_montecarolo_llr_large}
\end{figure}

Derivation of LLR variance under the null is more complicated than its expected value and does not lead to a convenient closed from solution. The appendix provides detailed analysis of the variance including a numerical approximation method using Taylor series similar to what we did for the expected value of LLR and compares it against the  Monte-Carlo method used here to estimate the variance in large networks.

\subsection{Null Distribution of LLR through Simulations}
In this section, we investigate the null distribution of the log-likelihood ratio in synthetic networks and show how it varies by network density. The synthetic networks have 30,000 nodes divided into two equal-sized groups and are all generated by DCSBM models. For each given value of density, we generate a DCSBM model that matches that density in expectation and generate its LLR distribution under the null. The $\Omega$ parameter of the DCSBM model is chosen such that the total number of edges in the network matches the requested density in expectation and 70\% of edges are in-group (i.e. 35\% within each group) and the remaining 30\% are out-group (15\% in each direction). $\Theta$ parameters are randomly generated according to a truncated power law ranging from 1 to 30,000 with exponent -0.3 and then normalized so that they sum up to 1 within each group. This procedure would create significant degree heterogeneity in the DCSBM model. 

Figure \ref{fig:analytical_vs_montecarolo_llr_large} illustrates how the null distribution of the log-likelihood ratio varies by the network density. The distributions are constructed through Monte Carlo sampling from the DCSBM model. The bars correspond to two standard errors around the expected LLR. We make the following observations regarding the blue bars that correspond to LLR distribution from the true model. First, the expected value of LLR matches the chi-squared prediction by Wilks theorem when the network is dense. However as explained above, the expected value of LLR is larger than this classical prediction for sparse networks. This is due to the fact that in sparse networks, slight random variations in cross-group linking among the nodes can be incorrectly picked up as brokerage patterns due to the small number of such links. Thus in order to reject the null, the evidence for presence of brokerage must be stronger for sparse networks than dense ones. Second, even though we don't have an analytical confirmation, the variance of LLR seems to be stable across different densities. In particular, variance for all network densities is about 30,000, again matching the chi-squared prediction, and does not vary by more than 4\% from this value. Finally, the chi-squared prediction from Wilks theorem seems to fit the LLR distribution well for dense graphs even though the theorem does not technically apply. In fact, a chi-squared distribution with such large degrees of freedom should approach the asymptotically normal distribution of LLR as discussed above.

\subsection{LLR Inference}
Figure \ref{fig:analytical_vs_montecarolo_llr_large} also includes the null distribution of LLR constructed from an estimated rather than the true model. This resembles common scenarios in practice where the true model is not available, thus the LLR distribution should be constructed from the sample. The red bars correspond to the parametric bootstrap  \citep{davison1997bootstrap} constructed from a single sample taken from the DCSBM model. We first draw a network from the true DCSBM model corresponding to each density, obtain its estimated parameters in DCSBM from maximum likelihood, then repeatedly draw new networks from the estimated model and compute their $\hat{\lambda}$ to generate its distribution under the null.

The bootstrap distribution matches the true sampling distribution for dense networks. However as the network gets sparser, the bootstrap (and analytical) distribution  underestimates the true LLR distribution which could lead to high type I error. This happens due to the combination of two factors. First, in the sparse regime, the higher order terms in Taylor series expansion of each term in equation (\ref{eq:llr_exp}) become non-trivial and should be included when estimating the expected value of LLR.  Second, the parameter estimates and in particular node degree parameters ($\hat{\theta}_i^o$ and $\hat{\theta}_i^i$), are not consistent in the sparse regime. Thus, using them in equation (\ref{eq:llr_exp}) won't lead to a consistent estimator for expectation of LLR either.
The appendix provides more detail on the bias of LLR distribution generated from a sample.

The difficulty with estimating LLR distribution in sparse networks can also be explained by the ``effective sample size'' of networks. The effective data size in dense graphs is of order $O(n^2)$ and even though there are $O(n)$ parameters in DCSBM, the estimated model parameters would still be in the large data limit and consistent \citep{Krivitsky2015question,yan2014model}. However,
in sparse graphs, the effective samples size and number of parameters grow at the same rate of $O(n)$. As there is only $O(1)$ observations per each parameter, we won't have consistent estimators for DCSBM parameters. 
Thus for sparse networks, neither analytical nor bootstrap LLR distribution that are constructed using estimated parameters match the correct distribution using the true parameters. Fundamentally, the problem is that the distribution of LLR becomes dependent on the parameters in the sparse network regime, as opposed to the dense regime where expected value of LLR is a constant that only depends on network size and number of groups (theorem \ref{theorem:llr_exp}). Thus, the plug-in estimator for the expected value of LLR is not consistent and this makes inference and testing impossible for sparse networks. Most social networks fall in the sparse regime, since average degree of nodes remains fixed as more nodes are added to the network. For such networks, new techniques are needed to estimate LLR distribution given that the current model selection suffers from high type I error.




\section{Conclusion}
\label{sec:conclusion}

Network models are increasingly used to study various social phenomena ranging from segregation \citep{diprete2011segregation, henry2011emergence},  clustering \citep{handcock2077model} and homophily \citep{mccormick2015latent} to employment outcomes \citep{jackson2004}. All such phenomena are either directly or indirectly related to biases in link formation in networks. Degree-Corrected Stochastic Block Model (DCSBM) is a random network model for estimating such biases and detecting the communities that arise from it. While DCSBM is successful in detecting communities and capturing homophily, it does not generate networks that match higher order homophily of the observed network. In this paper, we argue that matching higher order assortativities is important in social networks if we are concerned about the extent of (unequal) diffusion from one group to another and show empirically, based on a collection of school networks, that DCSBM significantly over-estimates the number of paths of length 2 or 3 between groups in social networks. We attribute this to unequal propensity in forming cross-group edges between members of a group, a phenomena referred to as brokerage in social network literature. brokers act as a bottleneck and networks with such nodes will have fewer paths between groups than networks whose cross-group edges are distributed more equally. We present a model based on DCSBM whose generated assortativity on paths of length 1 and 2 is consistent with the observed network. Even though the model does not make any guarantees in terms of assortativity on longer paths, we show empirically that the generated assortativity on paths of length 3 by our model is significantly more accurate than DCSBM.
This suggests that perhaps the most important factor behind unequal diffusion is simply the variation in the number of cross-group edges, which is fully accounted for in our mixed-propensity model.

Finally, we address the goodness of fit for our model versus DCSBM that does not account for brokerage. We characterize the distribution of the log likelihood ratio statistic and show that it is asymptotically normal. Even though the classical chi-squared distribution does not apply due to increasing number of parameters, the asymptotic distribution of LLR does match Wilks theorem predictions, but only in the dense network regime. This makes inference possible for dense graphs as the LLR null distribution does not depend on the true parameters. We show that LLR is still asymptotically normal with sparse networks, however it has a slightly larger mean to account for higher potential of overfitting. More importantly, the mean of LLR for sparse graphs depends on the unknown true parameters. We show analytically and empirically that a plug-in estimator will underestimate the LLR distribution for sparse networks since the maximum likelihood estimator of DCSBM parameters is not consistent for sparse graphs. Effectively, this makes model selection inference with plug-in estimators impossible in the sparse regime. This is particularly inconvenient as most social networks fall in the sparse regime due to the limited number of connections each individual can maintain. However, it may be possible to derive a consistent estimator for the LLR distribution in sparse networks, since its mean and variance solely depend on an aggregate function of model parameters. While the estimator for each parameter is inconsistent, it may be possible to develop a consistent estimator for their aggregate function. We believe this technique will be useful in other applications related to sparse networks and leave this topic as future work.

\section{Acknowledgements}
\label{sec:acknowledgements}
EJ was partly supported by NSF fellowship. This material is based upon work supported by the National Science Foundation Graduate Research Fellowship under
Grant No. 1122374. Any opinion, findings, and conclusions or recommendations expressed in this material are those of the authors(s) and do not necessarily reflect the views of the National Science Foundation.

\bibliographystyle{apalike}

\bibliography{references}

\section{Appendix}
\label{sec:appendix}

\subsection{Frequency Diffusion Paths Under MLE Model}
In this section, we provide extra analysis on the bias of model generated paths of length 2.
\subsubsection{Self-Loops:} In the main text, we assume that the observed networks do not have self-loops, even though the model allows for it and can certainly generate networks with self-loops. While we assume the first-order observed network does not have self-loops, its higher order networks do (imagine paths of length 2 that start with and end in the same node) and counting them is necessary to obtain an unbiased estimate of diffusion paths. Equation \ref{eq:diff_group_pl2} shows that expected number of the model generated paths of length 2 between any two different groups is the same as the value in the observed network. Similarly, \ref{eq:same_group_pl2} shows that if the observed network does not have any self-loops, the expected number of paths of length 2 between nodes of the same group matches the observed network.

The presence of self-loops has no effect on the observed or expected number of paths of length 2 between two distinct groups.
However, the presence of self-loops in the observed network leads to a positive bias in the number of in-group paths and consequently the estimated higher order assortativity as we show below.
In the main text, we assumed the number of observed paths of length 2 within a group is $\sum_{j} d_{j,r}^{i} d_{j,r}^{o}$. However in presence of self-loops, this value becomes:
\begin{align}
\begin{split}
    P_{rr}^{(2)} &= \sum_{j} \sum_{\substack{i,k \in r \\ i \neq k}} A_{ij} A_{jk} + \sum_{j} \sum_{\substack{i \in r \\ i \neq j}} A_{ij} A_{ji} + \sum_{j \in r} A_{jj} (A_{jj} - 1) \\
    &= \sum_{j} \sum_{i,k \in r } A_{ij} A_{jk} - \sum_{j \in r} A_{jj} \\
    &= \sum_{j} d_{j,r}^{i} d_{j,r}^{o} - \sum_{j \in r} A_{jj} 
\end{split}
\label{eq:same_group_pl2_selfloop}
\end{align}
Comparing equation (\ref{eq:same_group_pl2}) of the main text with  equation (\ref{eq:same_group_pl2_selfloop}) above indicates that expected number of in-group paths of length 2 generated by the MLE model has a positive bias, equal to the number of self-loops in the group,  when compared against the  corresponding observed value in presence of self-loops. This also implies that model generated assortativity on paths of length 2 will be higher than the observed assortativity in finite networks. However, the size of this bias compared the total number of in-group paths vanishes as the network grows larger and assortativity is nevertheless consistent as shown in the main text.

\subsection{Variance of the Log-Likelihood Ratio}
The main text establishes the asymptotic normality of LLR under DCSBM as the true model and derives it expected value. We develop the variance of LLR in this section and introduce an approximation method based on Taylor series expansion similar to what we used for its expected value. In contrast to the expected value of LLR, we do not obtain a convenient closed from expression for variance of LLR analytically and instead suggest to estimate it empirically.

Variance of LLR becomes complicated since the covariance between many terms, for example out-degree from a group and the the number of edges from that group to another, is non-zero. Monte Carlo methods such as parametric bootstrap would be an attractive alternative to estimate the variance of LLR. In fact, the LLR variance estimates in the main text are obtained through Monte Carlo using the true model parameters. In this section, we provide an analytical expression for variance and compare it against the estimates from Monte Carlo. The analytical method is computationally intensive, thus we conduct the comparison on moderate sized networks with 500 nodes. Using the LLR in equation (\ref{eq:llr}) and the expected value of its terms in equation (\ref{eq:poisson_expectations}), we can derive the variance of the LLR when DCSBM as the null is the true model. We will use the following auxiliary functions in the expression for variance to make it more readable:
\begin{align*}
& a(\mu) = \text{var}(X \log X) \\
& b(\mu, \lambda) = \text{cov}(X \log X, (X+U) \log(X+U)) \\
& c(\mu, \lambda, \gamma) = \text{cov}((X+U) \log(X+U), (X+W) \log(X+W)) \\
& \text{when} \quad\quad X \sim \text{Poisson}(\mu), \quad X+U \sim \text{Poisson}(\lambda), \quad X+W \sim \text{Poisson}(\gamma)
\end{align*}

\begin{align}
\begin{split}
\text{Var}(\hat{\lambda}) = \sum_{r,s} \sum_{i \in r} & \Big[  a(\theta_i^o \omega_{rs}) + a(\theta_i^i \omega_{sr}) \\
& - 4b(\theta_i^o \omega_{rs}, \omega_{rs}) - 4b(\theta_i^i \omega_{sr}, \omega_{sr})   \\
& - 2b(\theta_i^o \omega_{rs}, \theta_i^o  \textstyle \sum_g \omega_{rg}) - 2b(\theta_i^i \omega_{sr}, \theta_i^i \textstyle \sum_g \omega_{gr}) \\
& + 2b(\theta_i^o \omega_{rs}, \textstyle  \sum_g \omega_{rg}) + 2b(\theta_i^i \omega_{sr}, \textstyle \sum_g \omega_{gr}) \\
& + 2b(\theta_i^o \omega_{rs}, \textstyle  \sum_g \omega_{gs}) + 2b(\theta_i^i \omega_{sr}, \textstyle \sum_g \omega_{sg}) \\
& + 4c(\omega_{rs}, \theta_i^o \textstyle \sum_g \omega_{rg}, \theta_i^o \omega_{rs}) + 4c(\omega_{sr}, \theta_i^i \textstyle \sum_g \omega_{gr}, \theta_i^i \omega_{sr}) \\
& - 2c(\textstyle \sum_g \omega_{gs}, \theta_i^o \textstyle \sum_g \omega_{rg}, \theta_i^o \omega_{rs}) - 2c(\textstyle \sum_g \omega_{sg}, \theta_i^i \textstyle \sum_g \omega_{gr}, \theta_i^i \omega_{sr}) \Big] \\
+ \sum_{r,s} & \Big[ 4a(\omega_{rs}) \\
& - 4b(\omega_{rs}, \textstyle \sum_g \omega_{rg}) - 4b(\omega_{rs}, \textstyle \sum_g \omega_{gs}) \\
& + 2c(\textstyle \sum_g \omega_{rg}, \textstyle \sum_g \omega_{gs}, \omega_{rs}) \Big] \\
+ \sum_{r} \sum_{i \in r} & \Big[ a(\theta_i^o \textstyle \sum_g \omega_{rg}) + a(\theta_i^i \textstyle \sum_g \omega_{gr}) \\
& - 2b(\theta_i^o \textstyle \sum_g \omega_{rg}, \textstyle \sum_g \omega_{rg}) - 2b(\theta_i^i \textstyle \sum_g \omega_{gr}, \textstyle \sum_g \omega_{gr}) \Big] \\
+ \sum_{r} & \Big[ a(\textstyle \sum_g \omega_{rg}) + a(\textstyle \sum_g \omega_{gr}) \Big] \\
+ \sum_{r,s} \sum_{\substack{i \in r \\ j \in s}} & \Big[ 2c(\theta_i^o \omega_{rs}, \theta_j^i \omega_{rs}, \theta_i^o \theta_j^i \omega_{rs}) \\
& - 2c(\theta_i^o \omega_{rs}, \theta_j^i \textstyle \sum_g \omega_{gs}, \theta_i^o \theta_j^i \omega_{rs}) \\ 
& - 2c(\theta_i^i \omega_{sr}, \theta_j^o \textstyle \sum_g \omega_{sg}, \theta_i^i \theta_j^o \omega_{sr}) \\ 
& + 2c(\theta_i^o \textstyle \sum_g \omega_{rg}, \theta_j^i \textstyle \sum_g \omega_{gs}, \theta_i^i \theta_j^o \omega_{sr}) \Big]
\label{eq:llr_full_variance}
\end{split}
\end{align}

The variance expression above can not be easily converted to a more convenient form. Instead, we can compute it numerically. Given the true model parameters, we can either compute each term numerically or use approximations based on Taylor series expansions of function $a(\mu), b(\mu, \lambda), c(\mu, \lambda, \gamma)$. The former approach leads to an exact value for the variance, however it will be computationally intensive to compute all variance and covariance terms especially in large networks. The latter approach is computationally feasible but can suffer from inaccuracies if the assumptions behind approximations are not valid. The approximation to the variance relies on the following results based on Taylor series expansions:
\begin{align*}
& \text{cov}(X, X \log X) = \mu \log \mu + \mu - {1}{6\mu} \\
& \text{when} \quad X \sim \text{Poisson}(\mu) \quad \& \quad \mu \gg 1 \\
& a(\mu) = \mu \log^2(\mu) + 2\mu \log(\mu) + \mu + \frac{1}{2} + \frac{7 \log (\mu)}{15 \mu} - \frac{1}{6\mu} + \frac{\log (\mu)}{\mu^2} - \frac{13}{144 \mu^2} \\
& \text{when} \quad \mu > 1\\
& b(\mu, \lambda) =  (1 + E[\log U]) \text{cov}(X, X \log X)  \approx (1 + \log \lambda)(\mu \log \mu + \mu -  \frac{1}{6\mu}) \\
& \text{when} \quad \lambda \gg \mu \quad \& \quad \lambda \gg 1 \\
& c(\mu, \lambda, \gamma) = \text{var(X)} \Big[ E[\log U] E[\log W] + E[\log U] + E[\log W] +1 \Big] \\
& \quad\quad\quad\quad \approx \mu \Big[ \log \lambda \log \gamma + \log \lambda + \log \gamma + 1 \Big] \\
& \text{when} \quad \lambda \gg \mu \quad \& \quad \gamma \gg \mu 
\end{align*}

\begin{figure}[t]
\centering
\includegraphics[width=0.8\textwidth]{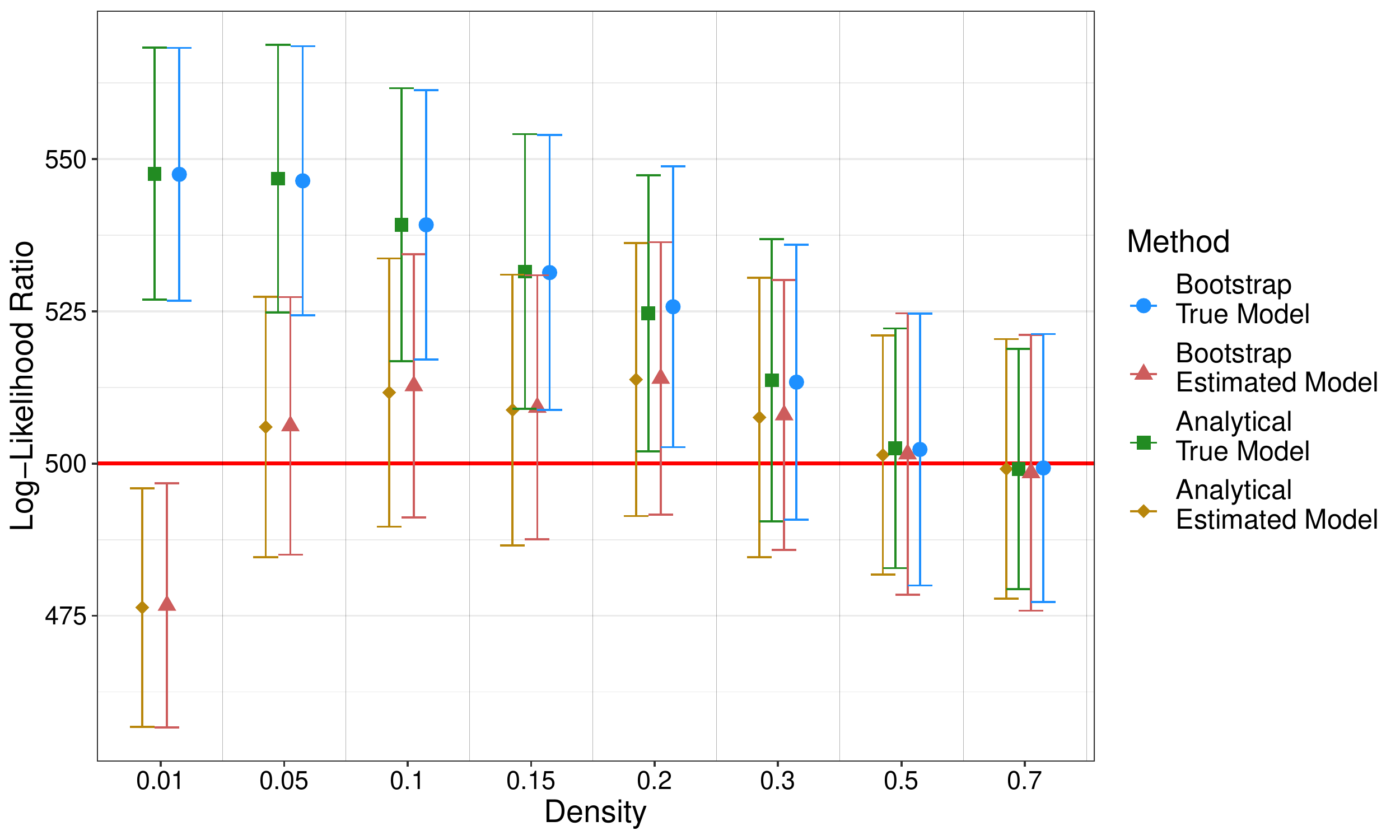}
\caption{Comparison of the analytical vs Monte Carlo methods for generating the distribution of LLR in networks with 500 nodes and varying density. Networks are constructed from a true DCSBM model. True (estimated) model distributions are generated by either sampling from the true (estimated) model or using the true (estimated) parameters in the analytical expression. Bars correspond to one standard error.}
\label{fig:analytical_vs_montecarolo_llr}
\end{figure}
The assumptions behind the approximations are too strict to be valid for common networks. Thus, to validate the Monte-Carlo estimation of variance, we rely on exact numerical values for each variance and covariance term in equation (\ref{eq:llr_full_variance}).
Figure \ref{fig:analytical_vs_montecarolo_llr} compares the analytical distribution of LLR using equations (\ref{eq:llr_exp}) and (\ref{eq:llr_full_variance}) versus those obtained through Monte-Carlo. The comparison is made across DCSBM-generated networks with varying level of density. The generation of network used the same method as referred to in the main text. We can make two main observations from figure \ref{fig:analytical_vs_montecarolo_llr}. First, the expected value and variance estimates of LLR based on a Monte-Carlo that uses the true parameter values for resampling is accurate and close to the values obtained analytically. Second, both the analytical and Monte-Carlo methods that use the estimated parameters depart from the true distribution as the network gets sparser. This issue was explained in the main text (further elaborated below) and was attributed to the lack of consistency in estimation of DCSBM parameters when the network is sparse.

\subsection{Estimating LLR distribution in sparse networks}
In the main text, we illustrated how the distribution of the log-likelihood ratio under the null (DCSBM) constructed from estimated model parameters departs from its true distribution when the network is in the sparse regime. We defined network sparsity in terms of node degree to each group. In particular, a network is considered sparse if there is at least one node whose expected degree to one group remains $O(1)$ as network size grows. This section characterizes the bias in the expected value of LLR under the null, if estimated using a plug-in estimator in equation (\ref{eq:llr_exp}). $\widehat{\EE{\hat{\lambda}}}$ below denotes the estimated LLR expected value using the model estimates where the expectation is taken over both sampled networks from the true model to obtain the estimated model first and then over resamples from the fitted model. In other words, we have $\widehat{\EE{\hat{\lambda}}} = \EE[\hat{\Theta}, \hat{\Omega}]{{\EE{\hat{\lambda} | \hat{\Theta}, \hat{\Omega}}}}$.
\begin{align*}
    \EE{\hat{\lambda}} - \widehat{\EE{\hat{\lambda}}}
    &= \begin{aligned}[t]
        \sum_{\substack{i \in N \\ g \in G}} \Big[ & f(\theta_i^o \omega_{g_i g}) - \EE{f(\hat{\theta}_i^o \hat{\omega}_{g_i g})} \\
        + &  f(\theta_i^i \omega_{g g_i}) - \EE{f(\hat{\theta}_i^i \hat{\omega}_{g g_i})} \Big] 
    \end{aligned} \\
    &- \begin{aligned}[t]
        \sum_{i \in N} \Big[ & f(\theta_i^o \textstyle\sum_{g \in G} \omega_{g_i g})  - \EE{f(\hat{\theta}_i^o \textstyle\sum_{g \in G} \hat{\omega}_{g_i g})} \\
        + &  f(\theta_i^i \textstyle\sum_{g \in G} \omega_{g g_i}) - \EE{f(\hat{\theta}_i^i \textstyle\sum_{g \in G} \hat{\omega}_{g g_i})} \Big] 
    \end{aligned} \\
    &- \sum_{r,s \in G} \Big[ 2f(\omega_{rs}) - 2\EE{f(\hat{\omega}_{rs})} \Big] \\
    &+ \begin{aligned}[t]
        \sum_{r \in G}\Big[ & f(\textstyle\sum_{s \in G} \omega_{rs}) - \EE{f(\textstyle\sum_{s \in G} \hat{\omega}_{rs})} \\
        + & f(\textstyle\sum_{s \in G} \omega_{sr}) - \EE{f(\textstyle\sum_{s \in G} \hat{\omega}_{sr})} \Big]
    \end{aligned}
\end{align*}
where $f(\mu) = \EE{X\log{X}}$ for $X \sim \text{Poisson}(\mu)$ as defined in the main text. We can replace each term above with its Taylor series expansion in equation (\ref{eq:taylor_series}) and note that the sum of first two terms in Taylor expansions lead to $ (|G| - 1)(|N| - |G|)$ for both true and estimated parameters as shown in Theorem \ref{theorem:llr_exp}. Thus, we only need to keep track of the difference between higher order terms. For simplicity, we only account for $O(\frac{1}{\mu})$ terms of the Taylor series expansion below which would be justified if $\mu > 1$. The analysis of the bias will not change if higher order terms are also included.
\begin{align}
\begin{split}
    \EE{\hat{\lambda}} - \widehat{\EE{\hat{\lambda}}}
    &\approx b_1 + b_2 \\
    b_1 &=
    \begin{aligned}[t]
        \frac{1}{12}\sum_{\substack{i \in N \\ g \in G}} \Bigg[ &\frac{1}{\theta_i^o \omega_{g_i g}} - \EE{\frac{1}{\hat{\theta}_i^o \hat{\omega}_{g_i g}}} \\
        + &  \frac{1}{\theta_i^i \omega_{g g_i}} - \EE{\frac{1}{\hat{\theta}_i^i \hat{\omega}_{g g_i}}} \Bigg] \\
    \end{aligned} \\
    &- \begin{aligned}[t]
        \frac{1}{12}\sum_{i \in N} \Bigg[ & \frac{1}{\theta_i^o \textstyle\sum_{g \in G} \omega_{g_i g}}  - \EE{\frac{1}{\hat{\theta}_i^o \textstyle\sum_{g \in G} \hat{\omega}_{g_i g}}} \\
        + &  \frac{1}{\theta_i^i \textstyle\sum_{g \in G} \omega_{g g_i}} - \EE{\frac{1}{\hat{\theta}_i^i \textstyle\sum_{g \in G} \hat{\omega}_{g g_i}}} \Bigg] \\
    \end{aligned} \\
    b_2 &= 
    \begin{aligned}[t]
        \frac{1}{12}\sum_{r \in G}\Bigg[ & \frac{1}{\textstyle\sum_{s \in G} \omega_{rs}} - \EE{\frac{1}{\textstyle\sum_{s \in G} \hat{\omega}_{rs}}} \\
        + & \frac{1}{\textstyle\sum_{s \in G} \omega_{sr}} - \EE{\frac{1}{\textstyle\sum_{s \in G} \hat{\omega}_{sr}}} \Bigg] \\
    \end{aligned} \\
    &- \begin{aligned}[t]
        \frac{1}{12}\sum_{r,s \in G} \Bigg[ 2\frac{1}{\omega_{rs}} - 2\EE{\frac{1}{\hat{\omega}_{rs}}} \Bigg] 
    \end{aligned}
\end{split}
\label{eq:llr_bias_anlaysis}
\end{align}
where we have divided the bias into two terms $b_1$ and $b_2$ and used the approximation instead of inequality since we have only accounted for the higher order terms of the Taylor expansion. First, we note that several of the difference terms above are negative according to the Jensen inequality. Thus the bias will generally be non-zero for sparse networks. Both figures \ref{fig:analytical_vs_montecarolo_llr_large} and \ref{fig:analytical_vs_montecarolo_llr} imply that the overall bias is zero for dense networks and positive for sparse networks. Close examination of the bias expression in equation (\ref{eq:llr_bias_anlaysis}) also confirms that both $b_1$ and $b_2$ are positive.
First, we observe that
\begin{align}
\begin{split}
\forall r,s \in G \quad\quad \frac{1}{\omega_{rs}} - \EE{\frac{1}{\hat{\omega}_{rs}}} < \frac{1}{\textstyle\sum_{s \in G} \omega_{rs}} - \EE{\frac{1}{\textstyle\sum_{s \in G} \hat{\omega}_{rs}}} \\
\forall r,s \in G \quad\quad \frac{1}{\omega_{sr}} - \EE{\frac{1}{\hat{\omega}_{sr}}} < \frac{1}{\textstyle\sum_{s \in G} \omega_{sr}} - \EE{\frac{1}{\textstyle\sum_{s \in G} \hat{\omega}_{sr}}}
\end{split}
\label{eq:llr_bias_anlaysis2}
\end{align}
since the Jensen gaps are larger in magnitude when the Poisson random variable in the denominator has a lower expected value. Inequalities (\ref{eq:llr_bias_anlaysis2}) imply that
\begin{align}
b_2 > 0
\end{align}
To evaluate $b_1$, we start by examining the expectations in equation (\ref{eq:llr_bias_anlaysis}).
\begin{align}
\begin{split}
\EE{\frac{1}{\hat{\theta}_i^o \textstyle\sum_{g \in G} \hat{\omega}_{g_i g}}} -
\EE{\frac{1}{\hat{\theta}_i^o \hat{\omega}_{g_i g}}} &= 
\EE{\frac{d_{g_i}^o}{d_i^o\textstyle\sum_{g \in G} m_{g_i g}}} - 
\EE{\frac{d_{g_i}^o}{d_i^o m_{g_i g}}} \\
&= \EE{\frac{1}{d_i^o}} -
\EE{\frac{\textstyle\sum_{s \in G} m_{g_is}^o}{d_i^o m_{g_i g}}} \\
&= \EE{\frac{1}{d_i^o}} -
\EE{\frac{1}{d_i^o} + \frac{\textstyle\sum_{s \neq g} m_{g_is}^o}{d_i^o m_{g_i g}}} \\
&= \EE{\frac{\textstyle\sum_{s \neq g} m_{g_is}^o}{d_i^o m_{g_i g}}} \\
&> 0
\end{split}
\label{eq:llr_bias_anlaysis2}
\end{align}
where the estimators correspond to DCSBM maximum likelihood, $d_{g_i}^o = \sum_{g \in G} m_{g_i g}$ is the total out-degree of all nodes in $g_i$ or the group node $i$ belongs to, and $d_i^o$ is the total out-degree of node $i$. There is a similar result for incoming edges:
\begin{align}
\begin{split}
\EE{\frac{1}{\hat{\theta}_i^i \textstyle\sum_{g \in G} \hat{\omega}_{g g_i}}} -
\EE{\frac{1}{\hat{\theta}_i^i \hat{\omega}_{g g_i}}} &= \EE{\frac{\textstyle\sum_{s \neq g} m_{s g_i}^i}{d_i^i m_{g g_i}}} \\
&> 0
\end{split}
\label{eq:llr_bias_anlaysis3}
\end{align}
Using equations (\ref{eq:llr_bias_anlaysis2}) and (\ref{eq:llr_bias_anlaysis3}) in evaluation of $b_1$ in (\ref{eq:llr_bias_anlaysis}), we conclude
\begin{align}
b_1 > 0
\end{align}
Thus, we have established that LLR expected value estimated through a plug-in estimator with model parameters underestimates true LLR expected value, leading to potential high type I error rate.
\begin{align}
    \EE{\hat{\lambda}} - \widehat{\EE{\hat{\lambda}}} > 0
\end{align}

\end{document}